\makeatletter
\def\UTFActivateAnd#1{\bgroup\def\UTFviii@defined##1{\egroup#1##1}}
\UTFActivateAnd\def定{\UTFActivateAnd\def}定令{\UTFActivateAnd\let}

\documentclass{IEEEtran}

\usepackage[english]{babel}

\usepackage[OT2,T1]{fontenc}
	\DeclareSymbolFont{CyrillicLetters}{OT2}{wncyr}{m}{n}
	\DeclareMathSymbol\Sha\mathalpha{CyrillicLetters}{"58}

\usepackage{mathtools,amssymb,bm,amsthm}
	\allowdisplaybreaks
	令Γ\Gamma 令Δ\Delta 令Θ\Theta 令Λ\Lambda 令Ξ\Xi 令Π\Pi 令Σ\Sigma 令Υ\Upsilon 令Φ\Phi
	令Ψ\Psi 令Ω\Omega 令α\alpha 令β\beta 令γ\gamma 令δ\delta 令ε\varepsilon 令ζ\zeta
	令η\eta 令θ\theta 令ι\iota 令κ\kappa 令λ\lambda 令μ\mu 令ν\nu 令ξ\xi 令π\pi 令ρ\rho
	令ς\varsigma 令σ\sigma 令τ\tau 令υ\upsilon 令φ\varphi 令χ\chi 令ψ\psi 令ω\omega
	令ϑ\vartheta 令ϕ\phi 令ϖ\varpi 令Ϝ\Digamma 令ϝ\digamma 令ϰ\varkappa 令ϱ\varrho
	令ϴ\varTheta 令ϵ\epsilona
	定𝛂{\bm\alpha} 定𝛃{\bm\beta} 定𝛄{\bm\gamma} 定𝛅{\bm\delta}
	定𝛆{\bm\varepsilon} 定𝛇{\bm\zeta} 定𝛈{\bm\eta} 定𝛉{\bm\theta} 定𝛊{\bm\iota}
	定𝛋{\bm\kappa} 定𝛌{\bm\lambda} 定𝛍{\bm\mu} 定𝛎{\bm\nu} 定𝛏{\bm\xi}
	定𝛐{\bm\omicron} 定𝛑{\bm\pi} 定𝛒{\bm\rho} 定𝛓{\bm\varsigma} 定𝛔{\bm\sigma}
	定𝛕{\bm\tau} 定𝛖{\bm\upsilon} 定𝛗{\bm\varphi} 定𝛘{\bm\chi} 定𝛙{\bm\psi}
	定𝛚{\bm\omega} 定𝛛{\bm\partial} 定𝛜{\bm\epsilon} 定𝛡{\bm\varpi}
	定𝛝{\bm\vartheta} 定𝛞{\bm\varkappa} 定𝛟{\bm\phi} 定𝛠{\bm\varrho}
	定𝐛{\bm b} 定𝐜{\bm c} 定𝐯{\bm v} 定𝐱{\bm x} 定𝐲{\bm y} 
	定𝟎{\bm0} 定𝟏{\bm1}
	定𝔼{\mathbb E} 定𝔽{\mathbb F} 定ℕ{\mathbb N} 定ℙ{\mathbb P} 定ℝ{\mathbb R}
	定ℤ{\mathbb Z}
	定𝒞{\mathcal C} 定𝒟{\mathcal D} 定ℰ{\mathcal E} 定ℱ{\mathcal F} 定𝒢{\mathcal G}
	定𝒮{\mathcal S} 定𝒯{\mathcal T}
	令ℓ\ell 令∂\partial 令∇\nabla
	定（{\bigl(} 定）{\bigr)}定［{\bigl[} 定］{\bigr]} 定｛{\bigl\{} 定｝{\bigr\}}
	令√\sqrt 令⌈\lceil 令⌉\rceil 令⌊\lfloor 令⌋\rfloor 令⟨\langle 令⟩\rangle
	令｜\mid 定＼{\setminus\nolinebreak} 令：\colon
	令ˆ\hat 令˜\tilde 令¯\bar 令˘\breve 令˙\dot 令¨\ddot 令°\mathring 令ˇ\check
	令∏\prod 令∑\sum 令∫\int 令⋀\bigwedge 令⋁\bigvee 令⋂\bigcap 令⋃\bigcup 令⨁\bigoplus
	令⨂\bigotimes
	令±\pm 令·\cdot 令×\times 令÷\frac 令†\dagger 令•\bullet 令∘\circ 令∧\wedge 令∨\vee
	令∩\cap 令∪\cup 令⊕\oplus 令⊗\otimes 令⊙\odot 令⋆\star
	令¬\neg 令…\ldots 令∀\forall 令∁\complement 令∃\exists 令∞\infty 令⊤\top 令⊥\bot
	令⋯\cdots 令★\bigstar
	令←\gets 令→\to 令↔\leftrightarrow 令↖\nwarrow 令↗\nearrow 令↘\searrow 令↙\swarrow
	令↞\twoheadleftarrow 令↠\twoheadrightarrow 令↤\mapsfrom 令↦\mapsto
	令↩\hookleftarrow 令↪\hookrightarrow 令↾\upharpoonright 令∈\in 令∉\notin 令∋\ni
	令∼\sim 令≅\cong 令≈\approx 定≔{\coloneqq} 令≕\eqqcolon 令≠\neq 令≡\equiv
	定≢{\not\equiv} 令≤\leqslant 令≥\geqslant 令≪\ll 令≫\gg 令⊆\subseteq 令⊇\supseteq
	令⋮\vdots 令⋰\adots 令⋱\ddots 令⟂\perp 令♭\flat 令♮\natural 令♯\sharp
	令⟵\longleftarrow 令⟶\longrightarrow 令⟼\longmapsto
	定⁰{^0} 定¹{^1} 定²{^2} 定³{^3} 定⁴{^4} 定⁵{^5} 定⁶{^6} 定⁷{^7} 定⁸{^8} 定⁹{^9}
	定₀{_0} 定₁{_1} 定₂{_2} 定₃{_3} 定₄{_4} 定₅{_5} 定₆{_6} 定₇{_7} 定₈{_8} 定₉{_9}
	令㏒\log 令㏑\ln
	定℃{\,^∘\text{C}}
	定†#1†{\text{#1}}
	定¯#1¯{\overline{#1}}
	定⋰#1{\tikz[baseline=-axis_height]\pgfmathsetmacro\sec{sec(min(#10,90-#10))*.15}
		\draw[nodes={inner sep=0}]node{.}(-#10:\sec)node{.}(180-#10:\sec)node{.};}
	定⋯{⋰0} 定⋮{⋰9} 定⋱{⋰{45.}}
	令＾\stackrel
	\@namedef[{\begin{equation*}}
	\@namedef]{\end{equation*}}
	\def\bma#1{\begin{bmatrix}#1\end{bmatrix}}
	\def\sma#1{\begin{smallmatrix}#1\end{smallmatrix}}
	\def\cas#1{\begin{cases*}#1\end{cases*}}

\usepackage{tikz-cd,tikz-3dplot}
	\usetikzlibrary{backgrounds,shapes.geometric}
	定色#1!#2#3#4#5#6#7{\definecolor{#1}{HTML}{#2#3#4#5#6#7}}
	色PMS2767!182B49  色PMS3015!00629B    色PMS1245!C69214  色PMS116!FFCD00     
	色PMS3115!00C6D7  色PMS7490!6E963B    色PMS3945!F3E500  色PMS144!FC8900     
	色Black!000000    色CoolGray9!747678  色PMS401!B6B1A9   色Metallic871!84754E
	\let\oldpointxyz\pgfpointxyz
	\newdimen\pgf@z
	\def\pgfpointxyz#1#2#3{\oldpointxyz{#1-\camerax}{#2-\cameray}{#3-\cameraz}
		\pgfmathsetlength\pgf@z{\rcarot*\pgftemp@x+\rcbrot*\pgftemp@y+\rccrot*\pgftemp@z}
		\pgfmathsetlength\pgf@x{\pgf@x*\cameras/(\camerad-\pgf@z)}
		\pgfmathsetlength\pgf@y{\pgf@y*\cameras/(\camerad-\pgf@z)}}
	\tikzset{every picture/.style={line cap=round,line join=round}}
	
	\def\dss{\hskip9inminus9in}

\usepackage{pgfplotstable,booktabs,multirow}
	\pgfplotsset{
		compat/show suggested version=false,compat=1.16,
		width=8cm,height=6cm
	}
	\newcommand\linlin[2][]{\begin{axis}[#1]#2\end{axis}}

\usepackage{csquotes}

\usepackage[backend=bibtex,style=alphabetic]{biblatex}
\usepackage[unicode,pdfusetitle]{hyperref}
	\hypersetup{
		pdfsubject={Information Theory (cs.IT); 05B20; 15A80},
		pdfkeywords={group testing, PCR testing, tropical arithmetic},
		colorlinks,allcolors=PMS2767,
	}

\usepackage[hyphenbreaks]{breakurl}

\usepackage{cleveref}
	
	定名#1:#2~#3?#4 {\crefname{#1}{#2#3}{#2#4}}
	名section:Section~?s 名appendix:Appendix~?s  名figure:Figure~?s 名table:Table~?s
	定理#1:#2~#3?#4 {\newtheorem{#1}[allams]{#2#3}名#1:#2~#3?#4 }
	理thm:Theorem~?s 理cor:Corollar~y?ies 理lem:Lemma~?s 理pro:Proposition~?s
	\theoremstyle{definition} 理dfn:Definition~?s 理exa:Example~?s
	\theoremstyle{remark} 理cla:Claim~?s 理rem:Remark~?s 理axi:Assumption~?s
	定式#1:#2~#3?#4 {名#1:#2~#3?#4 \creflabelformat{#1}{\textup{(##2##1##3)}}}
	式equ:equation~?s 式ine:inequalit~y?ies 式for:formula~?s 式mat:matri~x?ces
	式min:minimum~?s 式cod:codeword~?s 
	定標#1:#2?{\label@in@display@optarg[#1]{#1:#2}}
	\AtBeginDocument{\def\label@in@display#1{\incr@eqnum\tag{\theequation}標#1?}}

\begin{document}

\title{
                             Tropical Group Testing
}
\author{
                                  Hsin-Po Wang
                                      and
                                  Ryan Gabrys
                                      and
                                Alexander Vardy%
\thanks{
                              The authors are with
                  University of California San Diego, CA, USA.
                           This work was supported by
                    NSF grants CCF-2107346 and CCF-1764104.
               Emails: \{hsw001, rgabrys, avardy\} @eng.ucsd.edu
}}
\maketitle

\advance\baselineskip0pt plus.25pt minus.125pt

\begin{abstract}\boldmath
	Polymerase chain reaction (PCR) testing is the gold standard for diagnosing
	COVID-19.  PCR amplifies the virus DNA 40 times to produce measurements of
	viral loads that span seven orders of magnitude.  Unfortunately, the outputs
	of these tests are imprecise and therefore quantitative group testing
	methods, which rely on precise measurements, are not applicable. Motivated
	by the ever-increasing demand to identify individuals infected with
	SARS-CoV-19, we propose a new model that leverages tropical arithmetic
	to characterize the PCR testing process. Our proposed framework, termed
	tropical group testing, overcomes existing limitations of quantitative group
	testing by allowing for imprecise test measurements. In many cases, some of
	which are highlighted in this work, tropical group testing is provably more
	powerful than traditional binary group testing in that it requires fewer
	tests than classical approaches, while additionally providing a mechanism to
	identify the viral load of each infected individual.  It is also empirically
	stronger than related works that have attempted to combine PCR, quantitative
	group testing, and compressed sensing.
\end{abstract}

\begin{figure*}
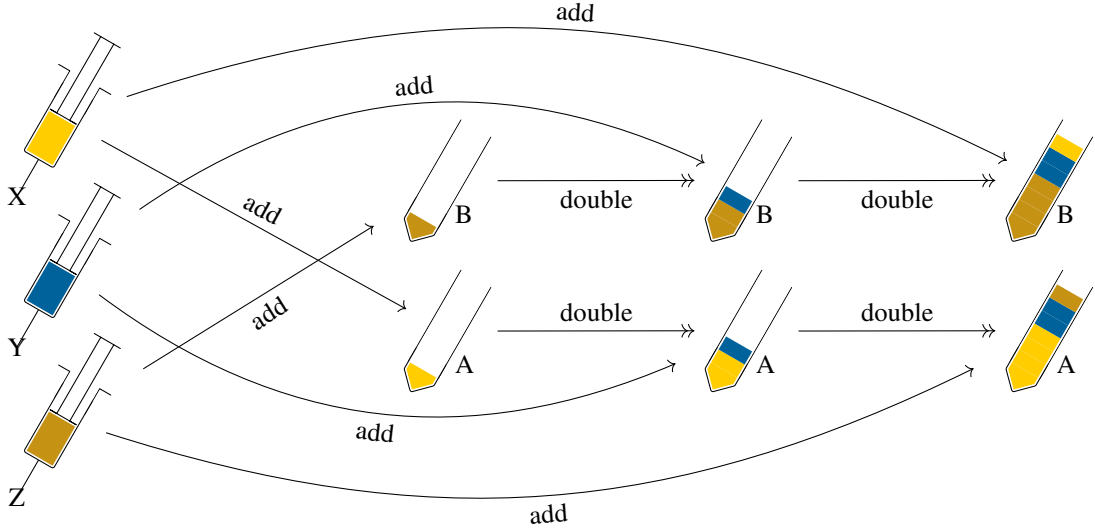

	\def\tubepop#1#2,{\xdef\car{#1}\xdef\cdr{#2}}
	\tikzset{
		syringe/.pic={
			\tikzset{scale=1/5,rotate=-30,name prefix/.get=\person}
			\draw(0,0)--(0,-3)node{\person};
			\draw(-2,6)-|(-1cm-4pt,4)(1cm+4pt,4)|-(2,6);
			\draw[,line width=2](-1,0)rectangle(1,5);
			\filldraw[white,line width=1.2](-1,0)rectangle(1,6);
			\draw(-1,9)--(1,9)(-1/2,9)|-(-1,3cm+4pt)--+(2,0)-|(1/2,9);
			\fill[\person](-1,0)rectangle(1,3);
			\draw(0,3)node()[transform shape,inner xsep=3cm,inner ysep=6cm]{};
		},
		tube/.pic={
			\tikzset{scale=1/5,rotate=-30}
			\draw(-1cm-4pt,1)--+(0,7)(1cm+4pt,1)--+(0,7);
			\draw[line width=2](-1,2)--(-1,1/2)--(0,-1/2)--(1,1/2)--(1,2)
				node[right,name prefix/.get=\tube]{\expandafter\tubepop\tube,\car};
			\filldraw[white,line width=1.2](-1,8)--(-1,1/2)--(0,-1/2)--(1,1/2)--(1,8);
			\expandafter\tubepop\contentstring{},
			\ifx\car\pgfutil@empty\else
				\fill[white,\car](-1,1)--(-1,1/2)--(0,-1/2)--(1,1/2)|-cycle;
			\fi
			\foreach\y in{1,...,7}{
				\expandafter\tubepop\cdr{},
				\ifx\car\pgfutil@empty\else\fill[white,\car](-1,\y)rectangle+(2,1);\fi
			}
			\draw(0,4)node()[transform shape,inner xsep=3cm,inner ysep=5cm]{};
		},
		X/.style={PMS116},
		Y/.style={PMS3015},
		Z/.style={PMS1245},
		T/.code={\def\contentstring{#1}},
	}
	$$\tikz{
		\draw(0,4)pic(X){syringe}(0,2)pic(Y){syringe}(0,0)pic(Z){syringe};
		\draw(5,3)pic(B1)[T=Z]{tube}(9,3)pic(B2)[T=ZZY]{tube}(13,3)pic(B3)[T=ZZZZYYX]{tube};
		\draw(5,1)pic(A1)[T=X]{tube}(9,1)pic(A2)[T=XXY]{tube}(13,1)pic(A3)[T=XXXXYYZ]{tube};
		\draw[every edge/.style={draw,->},nodes={sloped,auto}]
			(X)edge node{add}(A1)edge[bend left=22.5]node{add}(B3)
			(Y)edge[bend left=30]node{add}(B2)edge[bend right=30,']node{add}(A2)
			(Z)edge[']node{add}(B1)edge[bend right=22.5,']node{add}(A3)
		;
		\draw[every edge/.style={draw,->>},nodes={auto,'}]
			(B1)edge node{double}(B2)(B2)edge node{double}(B3)
			(A1)edge[']node{double}(A2)(A2)edge[']node{double}(A3)
		;
	}$$
	\caption{
		Syringes to the left:  Specimens are extracted from three people, X
		(gold), Y (blue) and Z (brown).  Tubes to the right:  Specimens will
		be remixed in two tubes, A and B, at the same time the PCR machine is
		amplifying DNA.  The schedule of mixing and amplifying is shown by the
		arrows.  We double the colored area to represent the fact that the virus
		DNA is replicated (in reality the liquid would not expand in volume).
		Assume one of X, Y, and Z is infected.  If the test results show that
		the DNA is twice as concentrated in tube A than in tube B, person X is
		infected.  If tubes A and B contain roughly the same concentration of
		DNA, person Y is infected.  Otherwise, person Z is infected.
	}\label{fig:syringe}
\end{figure*}

\section{Introduction}

	\IEEEPARstart{C}{OVID-19 pandemic} has highlighted the critical role that
	widely-accessible testing can have in controlling the spread of infectious
	diseases. Efficient testing schemes have the potential to simultaneously
	reduce the time to diagnosis while improving both the reliability and
	accuracy of the testing procedure. This subject has attracted significant
	attention in the open literature \cite{AE22}; however, existing works do not
	accurately model the semiquantitative information available at the output of
	the \emph{polymerase chain reaction} (PCR) testing methods used to detect
	the presence of SARS-CoV-19.
	
	PCR tests output \emph{cycle threshold} (Ct) values, which, as a result of
	the testing mechanism itself, are typically represented as semiquantitative
	measurements in the $\log$ domain \cite{FGSAHDTCMM20, ALLCYSAK20}). In this
	instance, the term \emph{quantitative} refers to the fact that the tests'
	readings are non-binary and \emph{semi} means that the readings are noisy or
	potentially inaccurate.  Previous semiquantitative approaches are ill-suited
	for modeling the output of a PCR test as previous works mostly rely on
	the assumption that test measurements are reported on a linear (rather
	than a $\log$) scale \cite{Perry20, DLO20, Austin20, GARPAGCGRGRG21,
	SLWSSOGSEGGMSFNSPH20, YMX20, YCWXM20, HN21, CSSEM21, PAJB21, PBJ20,
	LYLCC21, ZDLUAC21}

	As an illustration of the potential problem of modeling the PCR test
	outputs on a linear scale, consider the Ct value of a test as the dB
	value of a sound wave or the pH value of a liquid.  When adding a $50$~dB
	white noise with a $30$~dB one, we get a $50.04$~dB white noise that is
	indistinguishable from $50$~dB.  Due to the wide range in viral load between
	infected individuals, the same phenomenon for Ct values has been observed
	and is often referred to as \emph{masking} \cite{GPRRCGM21, BMR21, JM21}.  

	In order to address the masking issue and also to take advantage of
	the quantitative outputs available from PCR, we propose introducing
	\emph{delays} during the DNA amplification process in order to encode
	extra information. The basic idea will be to generate tests where each of
	the samples within a test can be inserted at different times. As a simple
	example of how this would work, suppose we design a test that consists of a
	single sample from an infected individual that has a Ct value of $X$. Then,
	if we delay inserting the sample by $Δ$ cycles, the output of the resulting
	test would be $X+Δ$.
	
	We use tropical multiplication $x⊙y≔x+y$ and tropical addition
	$x⊕y≔\min(x,y)$ to model the behavior of the Ct values that are provided as
	output of each of the PCR tests.  See \cref{tab:quant} for a comparison of
	this versus other models.  According to our model, the extra information can
	be retrieved by \emph{matching} the pattern of the Ct values against the
	pattern of the delays.  See \cref{fig:syringe} for an illustration of this
	process.  Delaying and matching, hand in hand, leave the masking issue
	nowhere to stand.

\pgfplotstableread{
	Regime				Reading							Remixing			
	Binary				{Negative, Positive}			$†Neg†∨†Pos†=†Pos†$	
	Tropical			$2^{-∞},2^{-40},\dotsc,2^{-12}$	$\min(30,15)=15$	
	Semiquantitative	$[0,3),[3,6),[6,9),\dotsc$		$[0,3)+[3,6)=[3,9)$	
	Quantitative		$0,1,2,3,4,5,\dotsc$			$8+9=17$			
}\tablequant
\begin{table}
	\caption{
		Four group testing regimes.  Blending specimens is modeled by
		logical or, minimum, interval addition, and ordinary addition.
	}\label{tab:quant}
	\def\arraystretch{1.5}
	$$\pgfplotstabletypeset[
		every head row/.style={before row=\toprule,after row=\midrule},
		every last row/.style={after row=\bottomrule},string type,
	]\tablequant$$
\end{table}

\subsection{Contributions of this paper}

	We propose a new model based upon two simple techniques\allowbreak
	---delaying and matching---for handling quantitative measurements that
	behave like dB values, pH values, or Ct values. With cleverly-crafted
	delays, we can attack a handful of scenarios as listed below.

	\begin{itemize}
		\item	When there is a single ($D=1$) infected person in a population
				of size $N$, we can identify her in just $T=2$ nonadaptive
				tests, where $N$ can be arbitrarily large (\cref{thm:difflay}).
				When delays are limited to $ℓ$ cycles, we show
				more generally $N≈Tℓ^{T-1}$ (\cref{thm:shell}).
		\item	For the case of $D=2$ infected persons within a population
				of size $N$, we give two constructions.
		\begin{itemize}
			\item	The first construction uses $T≈2√N$ nonadaptive tests
					(\cref{thm:pipet}).  In this construction,
					every person is present in only two tests.
			\item	The second construction uses $T≈㏒₂(N)+㏒₂(㏒₂N)$
					nonadaptive tests only (\cref{thm:modp}). When the delays
					are constrained, we can limit it to $ℓ≈㏒₂(N)/㏒₂(㏒₂N)$
					cycles, and achieve $T≈1.01㏒₂N$ (\cref{thm:bite}).
					This construction outperforms the information-theoretical
					bound of binary group testing.
		\end{itemize}
		\item	For general $D$, we give one necessary and two sufficient
				conditions of the existence for group testing schemes
				(\cref{sec:disjunct}).
		\item	When adaptive testing is allowed, $T=4$ tests are sufficient
				to find $D=2$ infected persons among arbitrarily many persons
				(\cref{thm:pad}).
		\item	In general, $T=3D+1$ adaptive tests are sufficient to locate
				$D$ infected persons among arbitrarily many persons
				(\cref{thm:nim}).  For this approach, one does not need to
				know $D$ beforehand.  When delays are limited to $ℓ$ cycles,
				we show that $T≈4D㏒_ℓN$ tests suffice (\cref{thm:deep}).
		\item	We present simulations that confirm that in many cases matching
				and delaying outperform related works (\cref{app:simulate}).
		\begin{itemize}
			\item	Without delaying, matching alone has a comparable
					sensitivity--specificity tradeoff as existing works.
			\item	The tradeoff improves as the range of the Ct
					value increases, i.e., matching works better when
					the masking issue is supposedly more destructive.
			\item	Adding random delay improves the tradeoff
					even further; the greater the variance of
					the delays, the better the performance.
		\end{itemize}
	\end{itemize}

	All construction proposed in this paper report the identities
	of the infected people as well as estimate the extent of infection;
	we are not sacrificing quantitativity in favor of sensitivity and
	specificity.  The reported Ct values can help further diagnoses
	\cite{PCMRZWNGHC20, WJCRDBACOOOSRH20}.
	
\subsection{Organization of this paper}
 
	\Cref{sec:prepare} reviews PCR and group testing and states
		the assumptions and goals of tropical group testing.
	\Cref{sec:diff-lay} develops optimal strategies
		to identify a single ($D=1$) infected person.
	\Cref{sec:pipet} considers the case of $D=2$ infected persons
		with minimum pipetting efforts.
	\Cref{sec:bite} presents a testing scheme for $D=2$ infected persons
		with nearly optimal number of tests.
	\Cref{sec:disjunct} presents necessary and sufficient conditions
		on non-adaptive testing strategies that hold for general $D$.
	\Cref{sec:pad} and \Cref{sec:nim} present our proposed adaptive strategies
		for $D=2$ and general $D$, respectively.
	For a detailed breakdown of the organization of the paper and the problems
	tackled in the respective sections, please see \cref{tab:class}.

\begin{figure*}
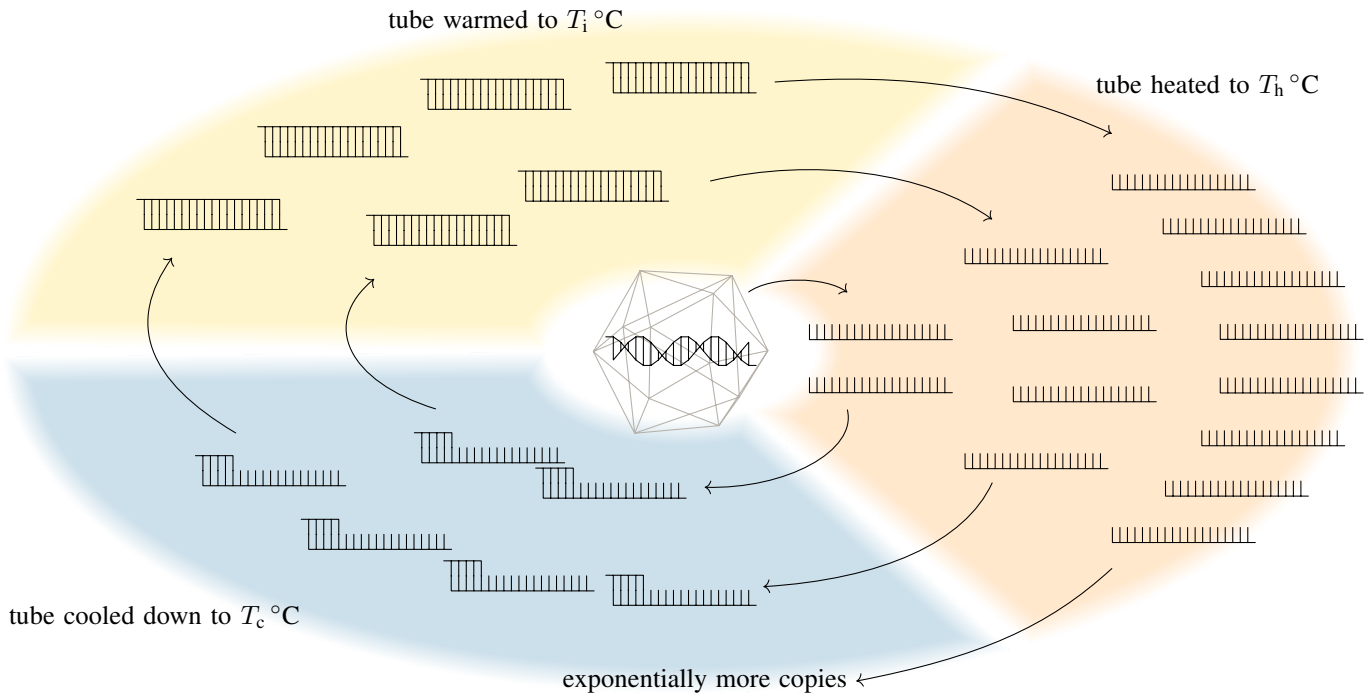

	\tdplotsetmaincoords{90-23.97565}{120.973882} 
	\def\camerax{0}\def\cameray{0}\def\cameraz{0}\def\cameras{3}\def\camerad{30}
	$$\tikz[yscale=5/11,xscale=9/10]{
		\def\r{2}\def\R{10}
		\def\a{59}\def\b{-59}
		\path[save path=\hot](\a:\r)arc(\a:\b:\r)--(\b:\R)arc(\b:\a:\R)--cycle;
		\def\a{179}\def\b{61}
		\path[save path=\med](\a:\r)arc(\a:\b:\r)--(\b:\R)arc(\b:\a:\R)--cycle;
		\def\a{299}\def\b{181}
		\path[save path=\cold](\a:\r)arc(\a:\b:\r)--(\b:\R)arc(\b:\a:\R)--cycle;
		\begin{scope}[every path/.style={use path=\hot},PMS144!20]
			\clip;
			\fill;
			\foreach\step in {100,95,...,0}{\draw[line width=\step/5,.!\step];}
		\end{scope}
		\begin{scope}[every path/.style={use path=\med},PMS116!20]
			\clip;
			\fill;
			\foreach\step in {100,95,...,0}{\draw[line width=\step/5,.!\step];}
		\end{scope}
		\begin{scope}[every path/.style={use path=\cold},PMS3015!20]
			\clip;
			\fill;
			\foreach\step in {100,95,...,0}{\draw[line width=\step/5,.!\step];}
		\end{scope}
		\draw[use path=\hot,white];
		\draw[use path=\med,white];
		\draw[use path=\cold,white];
		\draw[reset cm,tdplot_main_coords,PMS401]
			foreach\s in{+,-}{
				foreach\t in{+,-}{
					foreach\u in{+,-}{
						(\s9.87,0,\u6.1)--(0,\t6.1,\u9.87)--(\s6.1,\t9.87,0)--cycle
					}
				}
				(\s9.87,0,\s6.1)--(\s9.87,0,-\s6.1)
				(\s6.1,\s9.87,0)--(-\s6.1,\s9.87,0)
				(0,\s6.1,\s9.87)--(0,-\s6.1,\s9.87)
			}
		;
		\draw[reset cm,xscale=1/10,yscale=1/5]
			foreach\bp in{-9,...,9}{
				(\bp-1,{sin(\bp*36+36)})--(\bp,{sin(\bp*36+72)})--
				(\bp,{sin(\bp*36-72)})--(\bp+1,{sin(\bp*36-36)})
			}
		;
		\tikzset{
			single/.pic={
				\draw[scale=1/10]foreach\bp in{-9,...,9}{(\bp,0)|-+(1,-2)};
			},
			primer/.pic={
				\draw[scale=1/10]foreach\bp in{-9,...,-5}{(\bp,0)|-+(-1,2)}
								foreach\bp in{-9,...,9}{(\bp,0)|-+(1,-2)};
			},
			double/.pic={
				\draw[scale=1/10]foreach\bp in{-9,...,9}{(\bp,0)|-+(-1,2)}
								foreach\bp in{-9,...,9}{(\bp,0)|-+(1,-2)};
			}
		}
		\draw[->](60:2)to[bend left=55](35:3);
		\draw(15:3)pic{single}(-15:3)pic{single};
		\draw[->](-35:3)to[bend left=50](-85:4);
		\draw(-105:4)pic{primer}(-135:4)pic{primer};
		\draw[->](-155:4)to[bend left=45](155:5);
		\draw(135:5)pic{double}(105:5)pic{double};
		\draw[->](85:5)to[bend left=40](40:6);
		\draw(30:6)pic{single}(10:6)pic{single}(-10:6)pic{single}(-30:6)pic{single};
		\draw[->](-40:6)to[bend left=35](-80:7);
		\draw(-90:7)pic{primer}(-110:7)pic{primer}(-130:7)pic{primer}(-150:7)pic{primer};
		\draw[->](-160:7)to[bend left=30](160:8);
		\draw(150:8)pic{double}(130:8)pic{double}(110:8)pic{double}(90:8)pic{double};
		\draw[->](80:8)to[bend left=25](45:9);
		\draw(35:9)pic{single}(25.5:9)pic{single}(15:9)pic{single}(5:9)pic{single}
			(-5:9)pic{single}(-15:9)pic{single}(-25:9)pic{single}(-35:9)pic{single};
		\draw[->](-45:9)to[bend left=20](-75:10);
		\draw(105:10)node{tube warmed to $T_†i†℃$}
				(45:11)node{tube heated to $T_†h†℃$}
			(-135:11)node{tube cooled down to $T_†c†℃$}
				(-75:10)node[left]{exponentially more copies};
		\pgfresetboundingbox
		\path circle(\R);
	}$$
	\caption{
		An illustration of how PCR works.  The fan at 1--5 o'clock: the
		heat denaturalizes the double helixes from a DNA virus, yielding
		single-stranded DNA.  The fan at 5--9 o'clock: as the tube cools
		down, the primers bind to the single-stranded DNA.  The fan at
		9--1 o'clock: the tube is warmed to the working temperature of the
		polymerase; the polymerase synthesizes the complementary part of
		the DNA.  This procedure is repeated until the number of copies of
		DNA is detectable or we feel impatient and opt to stop, whichever is
		earlier.  For RNA viruses, a reverse-transcription step is inserted
		after the virus coat is broken and before DNA amplification starts.
	}\label{fig:PCR}
\end{figure*}

\section{Background and Problem Statement}\label{sec:prepare}

	In this section, we review the working principles of the PCR process.  After
	that, we review the taxonomy of group testing and provide a brief history of
	quantitative group testing.  Next, we introduce the notion of delays and how
	this notion is formalized using the tropical semiring.  Finally, we conclude
	this section with our problem statement.

\subsection{The PCR process}

	PCR stands for polymerase chain reaction.  It is performed by a machine that
	holds a collection of tubes, each tube containing some specimens.  A PCR
	machine first heats the tubes up to a temperature $T_†h†℃$;  the heat
	decouples every double-helical DNA molecule into two single-stranded DNA
	templates.  The machine then cools the tubes down to another temperature
	$T_†c†℃$; at this temperature, the primers (short, single-stranded DNA
	segments) will attach to the DNA templates, labeling the starting point of
	DNA replication.  Following that, the machine warms the tubes up to an
	intermediate temperature $T_†i†℃$, which is the working temperature for the
	polymerases, the enzymes that will be completing the single-stranded DNA
	templates to form double-stranded DNA molecules.  A \emph{cycle} means that
	the tubes undergo $T_†h†℃$, $T_†c†℃$, and $T_†i†℃$ once, which implies that
	the concentration of DNA is doubled.  Repeating this process for $r$ cycles
	increases the concentration of DNA by $2^r$-fold.  See \cref{fig:PCR} for an
	illustration.

	In a \emph{quantitative} PCR\footnote{ It is also called \emph{real-time}
	PCR but still abbreviated as qPCR, not RT-PCR.  RT-PCR refers to another
	technique called reverse transcription PCR.  The combination of the two
	is abbreviated as RT-qPCR, which is the kind that is used in this pandemic.}
	test, some fluorescent dyes are added into the tubes; these dye molecules
	emit light when attaching to DNA.  As the process of DNA-amplification
	continues, more and more of these fluorescent molecules will attach
	themselves to the newly-created DNA content.  Eventually, the tubes will
	emit sufficient light that triggers a sensor.  When this happens, the
	current cycle count is reported as the \emph{cycle threshold} (Ct) value of
	the specimen.  Accordingly,
	\[†Ct value†=⌊-㏒₂(†viral load†)+†constant†⌋.\]
	If all tubes turn out positive or some tubes take too long ($40$ cycles in
	practice) to emit a sufficient amount of light, the machine is turned off.
	For those tubes that did not trigger the sensor, negative results are
	reported.

	For fluorescence-to-cycle plots in real life, see \cite{WHMR97,
	GTVRRFMRGTAB20, PWRKOMSVHBSFL21, MAATWASHABB21}.  For statistics of Ct
	values, see \cite{JBMVSBBTSSKMSZHKSECD21, JKKLGLM21, JGJO20, BMR21,BHF21}
	For the relation between Ct values and dilution, see
	\cite{YATAMBSKGlSHMlHGSK20, LPBRGGBSS20, ABMHKI20,
	MNBSURNMRNNSMNNUMMMNTN21}.

\subsection{Group testing}

	Group testing was introduced by Dorfman \cite{Dorfman43} and has been of
	research interest since then.  For introductions of group testing, see Du
	and Huang's book \cite{DH93, DH99}, Ngo and Rudra's lecture notes
	\cite{NR11}, and Aldridge, Johnson, and Scarlett's survey \cite{AJS19}.  For
	the most recent works not covered by the survey, see \cite{BCSYZ21, BPS21,
	CGHL21, CGM21, GPRRCGM21, LPC21}.  Group testing comes with a variety of
	flavors.  The survey \cite[Section~1.1]{AJS19} provides the following
	categorization of existing group testing works:
	\begin{itemize}
		\item	Adaptive v.\ nonadaptive.
		\item	Zero error probability v.\ small error probability.
		\item	Exact recovery v.\ partial recovery.
		\item	Noiseless v.\ noisy testing.
		\item	Binary v.\ non-binary outcomes.
		\item	Combinatorial v.\ iid prior.
		\item	Known v.\ unknown number of defectives.
	\end{itemize}
	The classification of the various scenarios studied
	throughout this paper is listed in \cref{tab:class}.

	We next discuss group testing models that involve non-binary test outcomes,
	which involves quantitative group testing as well as our tropical group
	testing model.

	The earliest form of non-binary group testing dates back to the 1940s, when
	balance scales were evaluated as a tool to single out counterfeit coins
	\cite{Hwang87, GN95, khovanova13}.  A variant of this coin-weighing problem
	considers the usage of pointer scales to count how many counterfeit coins
	are in the queried pool.  This problem later became known as
	\emph{quantitative} group testing and many of the techniques and results
	have found application in, say, additive multiple access channels
	\cite{GHKL19, FL20}.  On a parallel track, \emph{threshold} group testing
	considers the setup where each test outputs a “yes” when the the number of
	counterfeit coins exceeds a certain threshold $θ$, and it outputs a “no”
	otherwise \cite{Damaschke06, Dyachkov13, Cheraghchi13}.
	
	A unification of quantitative and threshold group testing, which has
	received recent attention in \cite{EM14s, EM14g, EM16, CGM21, GPRRCGM21}
	due to its applications to genomic sequence processing, is known as
	\emph{semiquantitative} group testing.  In semiquantitative group 
	testing, the positive reals are partitioned into intervals
	\[[0,θ₁)\,,\,[θ₁,θ₂)\,,\,[θ₂,θ₃)\,,\,\dotsc\]
	and each test outputs a number that indicates in which interval the true
	value lies in.  Previous works considered the setup where the endpoints
	$\{θ_j\}_j$ represent an arithmetic progression.  As will be discussed later
	in the next subsection, motivated by the PCR testing method, we consider a
	semiquantitative setup where the endpoints $\{θ_j\}_j$ are a geometric
	progression.  See \cref{tab:quant} for a summary of existing
	semiquantitative group testing regimes.

\pgfplotstableread{
    { }                        III   IV--V VI    VII   VIII  A     B     C     D   
    nonadaptive~v.~adaptive    N     N     N     A     A     N     N     A     N   
    error~probability          $0$   $0$   $0$   $0$   $0$   $0$   $0$   $0$   $ε$ 
    exact~v.~partial~recover   E     E     E     E     E     E     E     E     P   
    noisy~v.~noiseless         ¬     ¬     ¬     ¬     ¬     ¬     ¬     ¬     ¬   
    tropical~v.~quantitative   T     T     T     T     T     T     T     T     Q   
    combinatorial~v.~iid       C     C     C     C     C     C     C     C     iid 
    known~\#~of~infected       $1$   $2$   $D$   $2$   ?     $1$   $2$   ?     ?   
    un-~v.~limited~delay       U     U     U     U     U     L     L     L     L   
}\tableclass
\begin{table}
	\caption{
		The classification of the sections of this paper per
		\cite[Section~1.1]{AJS19}'s six criteria.  We added one more
		criterion that is about whether delay can grow linear in $N$.
	}\label{tab:class}
	\def\arraystretch{1.2}
	\tabcolsep3pt
	定¬{less}
	$$\pgfplotstabletypeset[
		every head row/.style={
				before row={
					\toprule
					\multirow2*{Classification}
					&	\multicolumn5c{Section}
					&	\multicolumn4c{Appendix}	\\
				}
			,
			after row=\midrule},
		every last row/.style={after row=\bottomrule},string type,
	]\tableclass$$
\end{table}

\subsection{Delay and tropical semiring (aka min-plus algebra)}

	In this work, we consider a special group testing scenario whereby one
	adds specimens into each tube at different cycles.  Consider the following
	example procedure that is also illustrated in \cref{fig:syringe}:
	\begin{itemize}
		\item	First, we place empty tubes A and B into the PCR machine.
		\item	Next, we put specimens X and Z into tubes A and B, respectively.
		\item	Run the machine for $1$ cycle.
		\item	Put specimen Y into tubes A and B.
		\item	Run the machine for $1$ cycle.
		\item	Put specimens Z and X into tubes A and B, respectively.
		\item	And finally, run the machine for another $38$ cycles or
				until the detection of light signal, whichever is earlier.
	\end{itemize}
	Unlike traditional group testing---which is only concerned with how to
	distribute specimens into tubes---this paper is also concerned with the
	question of \emph{when} specimens should be placed into their respective
	tubes.

	In Figure~\ref{fig:syringe}, since specimen Y is put into the tubes after
	one cycle has elapsed, we say that specimen Y is  \emph{delayed by $1$
	cycle}.  More generally, if a specimen is inserted into a tube after $δ$
	cycles have elapsed, we say that the specimen is \emph{delayed by $δ$
	cycles}.  If a specimen is never put into a  tube, we say that the
	corresponding delay is infinity.

	We now recall a mathematical framework that will be used to describe the
	behavior of Ct values under pooling and delaying.  The real numbers with
	(positive) infinity, $ℝ∪\{∞\}$, with operators $⊕$ and $⊙$, as defined by
	\begin{align*}
		x⊕y	&	≔\min(x,y)	&	(†in particular †x⊕∞	&	≔x),	\\
		x⊙y	&	≔x+y		&	(†in particular †x⊙∞	&	≔∞),	
	\end{align*}
	is called the \emph{tropical semiring} or the \emph{min-plus algebra}%
	\footnote{ Sometimes $\min$ is replaces by $\max$ and the structure
	is called a \emph{max-plus} algebra.  The theory is exactly the same up
	to some sign changes.}.  The tropical semiring captures the behavior of
	addition and multiplication through the lens of $㏒_q$, for some large
	$q$, and finds applications in algebraic geometry and combinatorics
	\cite{MS15, Joswig22}.  As an example, the core of the Floyd--Warshall
	algorithm (ibid.\ or \cite{CLRS22}) is tropical matrix multiplication,
	$A⊙B$, whose $(i,k)$th entry is defined to be
	\[⨁_j(A_{ij}⊙B_{jk})=\min_j(A_{ij}+B_{jk}).\]
	As will be justified by simulations in \cref{app:simulate}, we use tropical
	arithmetic to model the Ct values of pooled and delayed specimens.

	\begin{axi}[Tropical model]
		The mixture of specimens with Ct values $x₁,x₂,\dotsc,x_N$, each
		delayed by $δ₁,δ₂,\dots,δ_N$ cycles, respectively, has Ct value
		\[⨁_j(δ_j⊙x_j)=\min_j(δ_j+x_j).\]
		When the context is clear, we write $𝛅⊙𝐱$ to denote a tropical
		vector--vector multiplication.  When there are multiple tests,
		we write $S⊙𝐱$ to denote a tropical matrix--vector multiplication.
		We call $S$ a \emph{schedule matrix}.
	\end{axi}

	With this assumption we can claim, formally, the goal of this paper.

\subsection{Problem statement}

	The goal of this paper is to construct tropical group testing
	schemes that allow us to diagnose a population efficiently.
	The following two definitions state our goal precisely.

	\begin{dfn}[Nonadaptive testing]
		A \emph{$(T,N,D)$-tropical code} is a schedule matrix
		$S∈(\{0\}∪ℕ∪\{∞\})^{T×N}$ such that, for any distinct vectors
		$𝐱,𝐲∈(\{0\}∪ℕ∪\{∞\})^{1×N}$, each having at most $D$ finite entries,
		\[S⊙𝐱≠S⊙𝐲.\]
		A tropical code is said to be \emph{within maximum delay $ℓ$}
		if $S∈\{0,1,\dotsc,ℓ,∞\}^{T×N}$.
	\end{dfn}

	\begin{dfn}[Adaptive testing]
		An \emph{$R$-$(T,N,D)$-tropical protocol} is a series of
		$R$ functions, $𝒮^{(1)},𝒮^{(2)},\dotsc,𝒮^{(R)}$, that feed
		on past results and output variable-height schedule matrices
		\begin{align*}
			S^{(1)}	&	=𝒮^{(1)}(),	\\
			S^{(2)}	&	=𝒮^{(2)}(S^{(1)}⊙𝐱),	\\
			S^{(3)}	&	=𝒮^{(3)}\left(\bma{S^{(1)}\\S^{(2)}}⊙𝐱\right),	\\
					&	\mkern10mu⋮	\\
			S^{(R)}	&	=𝒮^{(R)}\left(\bma{S^{(1)}\\⋮\\S^{(R-1)}}⊙𝐱\right)
		\end{align*}
		such that (i) the numbers of rows may depend on past
		results but the total is $≤T$,  (ii) the numbers
		of columns are $N$, and (iii) the final result
		\[\bma{
			S^{(1)}	\\
			⋮	\\
			S^{(R)}
		}⊙𝐱\]
		is unique among all $𝐱∈(\{0\}∪ℕ∪\{∞\})^N$ having at most
		$D$ finite entries.  A tropical protocol is said to be
		\emph{within maximum delay $ℓ$} if the schedule matrixes
		use only the alphabet $\{0,1,\dotsc,ℓ,∞\}$.
	\end{dfn}

	\begin{rem}
		Our setup is more demanding than the traditional group testing setup in
		that we will be requiring not only the identities of the infected
		individuals but also the extent of infection of each infected
		individual.
	\end{rem}

	\begin{rem}
		Our setup is also more resilient to the fuzziness innate to PCR testing
		in that (i) our decoder makes decisions based on integer (instead of
		real number) Ct values and (ii) every time two specimens are mixed
		together, the one with fewer virus particles are “forgotten”.  This
		harsh rule forces us to find workarounds that must not rely on the type
		of arithmetic that “subtracts $50$~dB from $50.04$~dB to obtain
		$30$~dB”.
	\end{rem}

	\begin{rem}
		The theorems proved in this paper hold true even if we replace the
		alphabet $\{0\}∪ℕ∪\{∞\}$ of the Ct values by $ℝ∪\{∞\}$.  This makes
		tropical group testing instantly adapted to scenarios with fractional Ct
		values.  In the latter scenarios, since tropical arithmetic retains the
		meaning after rescaling, the alphabet $\{0,1,\dotsc,ℓ\}$ of the delays
		shall be understood as the multiples of the minimum difference of the Ct
		values that we can tell apart.
	\end{rem}

	The next section is a warmup with the $D=1$ case.

\section{One Infection and Diff-Lays}\label{sec:diff-lay}

	This section considers the simplest case whereby at most one person is
	infected.  It demonstrates how tropical group testing can find this infected
	person using only two tests, regardless of how many people are taking the
	tests.

\subsection{Two tests on three people---a toy example}

	Let there be three persons X, Y, and Z and denote their Ct values by
	$x,y,z∈\{0\}∪ℕ∪\{∞\}$.  Let there be two tubes, A and B, and denote
	their Ct values by $a,b∈\{0\}∪ℕ∪\{∞\}$.  This subsection describes a
	$(2,3,1)$-tropical code.

	Encoding:  The mixtures of the specimens are prepared as below.
	\[\bma{
		a	\\	b
	}≔\bma{
		0	&	0	&	7	\\
		7	&	0	&	0	
	}⊙\bma{
		x	\\	y	\\	z
	}\label{mat:toy}\]
	That is,\footnote{ A reminder that this schedule
	is different from \cref{fig:syringe}.}
	\begin{itemize}
		\item	put tubes A and B into the PCR machine;
		\item	put specimens X and Y into tube A;
		\item	put specimens Y and Z into tube B;
		\item	run the machine for $7$ cycles;
		\item	put specimens Z and X into tubes A and B, respectively;
				and finally
		\item	run the machine for $33$ more cycles or until the light
				sensor reports detects the virus, whichever is earlier.
	\end{itemize}
	Instead of delaying by $7$ cycles, one can delay by any nonzero cycles,
	including infinity.  A larger delay is safer when Ct values are noisy;
	so one should consider maxing out the available delays.

\begin{figure}
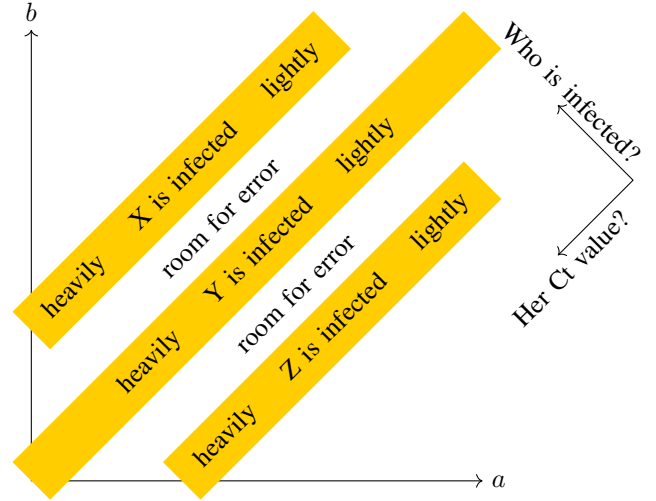

	$$\tikz{
		\draw[<->](0,6)node[above]{$b$}|-(6,0)node[right]{$a$};
		\draw[PMS116,line cap=butt,line width=7.070mm,nodes={black,sloped}]
			(0,2)--node{heavily~~~~~X is infected~~~~~lightly}+(4,4)
			(0,0)--node{heavily~~~~~Y is infected~~~~~lightly}+(6,6)
			(2,0)--node{heavily~~~~~Z is infected~~~~~lightly}+(4,4)
		;
		\path[nodes=sloped]
			(.5,1.5)--node{room for error}+(4,4)
			(1.5,.5)--node{room for error}+(4,4)
		;
		\draw[shift={(8,4)},<->]
			(-1,1)node[above,rotate=-45]{Who is infected?}--
			(0,0)--(-1,-1)node[below,rotate=45]{Her Ct value?}
		;
	}$$
	\caption{
		The configuration space of two PCR tests and how to decode the test
		results to infer (i) who is infected and (ii) how infected she is.
		The gaps between the belts are not necessary under the topical model
		but make rooms for noisy Ct values in real life.
	}\label{fig:triton}
\end{figure}

	Decoding:  If all three individuals are healthy, then both tubes will report
	negative results.  If X is infected, then we will see $a-b=-7$ and we can
	infer her Ct value $x=a$.  If Y is infected, we will see $a-b=0$ and infer
	her Ct value $y=a$.  Similarly, if Z is infected, we will see $a-b=7$ and
	infer $z=b$.  \Cref{fig:triton} picturizes how to interpret the test
	results.

	The next proposition summarizes the preceding discussion.

	\begin{pro}
		Schedule \cref{mat:toy} is a $(2,3,1)$-tropical code.
	\end{pro}

	The difference of delays, for instance $0-7=-7$ for specimen X, is called
	the \emph{diff-lay} and will play a central role in the sequel.  The next
	subsection generalizes the $(2,3,1)$-tropical code by generating many more
	diff-lays.

\subsection{Two tests on many people---general
	\texorpdfstring{$(2,N,1)$}{(2, N, 1)}-tropical code}

	One can define a $(2,4,1)$-tropical code as
	\[\bma{
		0	&	0	&	2	&	7	\\
		7	&	2	&	0	&	0
	}\]
	where the possible diff-lays are $a-b∈\{-7,-2,2,7\}$.
	One can also define a $(2,5,1)$-tropical code as
	\[\bma{
		0	&	0	&	0	&	4	&	7	\\
		7	&	4	&	0	&	0	&	0
	}\]
	where the possible diff-lays are $a-b∈\{-7,-4,0,4,7\}$.
	
	We can test as many people as we want so long as the delays can keep up.
	Say there is a cap on delays, which we denote by $ℓ$. Then it is possible
	and optimal to test $2ℓ+3$ persons at once via the following schedule
	matrix.
	\[\arraycolsep4pt\setcounter{MaxMatrixCols}{20}
	\bma{
		0  &  0  &  0  &⋯&  0  &  0  &  0  &  1  &  2  &⋯& ℓ-1 &  ℓ  &  ∞  \\
		∞  &  ℓ  & ℓ-1 &⋯&  2  &  1  &  0  &  0  &  0  &⋯&  0  &  0  &  0  
	}\label{mat:hinge}\]
	Notice how each column has a unique diff-lay and they exhaust all possible
	diff-lays from $-∞,-ℓ,-ℓ+1$ to $ℓ-1,ℓ,∞$.  Alternatively, one might prefer
	the following schedule matrix
	\[\arraycolsep4pt\setcounter{MaxMatrixCols}{20}
	\bma{
		0 & 0 & 1 &  1  &  2  & ⋯ & ℓ-1 & ℓ-1 & ℓ & ℓ & ∞\\ 
		∞ & ℓ & ℓ & ℓ-1 & ℓ-1 & ⋯ &  2  &  1  & 1 & 0 & 0
	}\label{mat:alter}\]
	because the pipetting works are spread out over time (at most four
	pipets per cycle).  The following theorem formalizes this idea.

	\begin{thm}[One patient, two tests]\label{thm:difflay}
		Any schedule matrix $S∈\{0,\dotsc,ℓ,∞\}^{2×N}$ that
		contains no infinite column $[\sma{∞\\∞}]$ and satisfies
		\[\bigl|\{S_{1j}-S_{2j}｜j∈[N]\}\bigr|=N\]
		is a $(2,N,1)$-tropical code within maximum delay $ℓ$.
		Every such $S$ must satisfy $N≤2ℓ+3$.
	\end{thm}

	\begin{proof}
		Let $𝐱∈(\{0\}∪ℕ∪\{∞\})^N$ have at most one finite entry.
		Iff $𝐱$ is entirely infinite (everyone is healthy),
		$S⊙𝐱$ will be entirely infinite (both tubes are negative).
		Next, suppose $x_j$ is finite for some $j$.  Let
		\[\bma{
			a	\\	b
		}≔S⊙𝐱.\]
		The diff-lay would then be $a-b=(S_{1j}+x_j)-(S_{2j}+x_j)=S_{1j}-S_{2j}$
		and by that diff-lay we can uniquely determine $j$.  Once $j$ is known,
		we can compute her Ct value $x_j=a-S_{1j}$ or $x_j=b-S_{2j}$.

		$N$ must be less than or equal to $2ℓ+3$ because all possible diff-lays
		are $-∞,-ℓ,-ℓ+1,\dotsc,ℓ-1,ℓ,∞$.  Without uniqueness of diff-lays,
		columns having the same diff-lay are interchangeable.
	\end{proof}

	\begin{cor}
		\Cref{mat:hinge,mat:alter} are both $(2,2ℓ+3,1)$-tropical codes
		within maximum delay $ℓ$ that attain the equality $N=2ℓ+3$.
	\end{cor}
	
	What is good about schedule \cref{mat:hinge} is that there are algebraic
	formulas that compute who is infected $j=a-b+ℓ+2$ and how sick she is
	$x_j=\min(a,b)$ (unless infinity involves).  For schedule \cref{mat:alter},
	the algebraic formulas read $j=a-b+ℓ+2$ and $x_j=⌊(a+b-ℓ)/2⌋$ (unless
	infinity involves).

	In \cref{thm:difflay}, $ℓ$ scales linearly in $N$.  As
	for how to handle the $D=1$ case when $ℓ≪N$, \cref{pf:shell}
	proves the following generalization of \cref{thm:difflay}.

	\begin{thm}[One patient]\label{thm:shell}
		There is a $(T,N,1)$-tropical code within
		maximum delay $ℓ$ iff $N≤(ℓ+2)^T-(ℓ+1)^T$.
	\end{thm}

	In the upcoming sections, we consider the setup
	where more than a single person can be infected.

\section{Two Infections and Minimum Pipetting}\label{sec:pipet}

	In the next two sections, we set the number of infected persons to be at
	most $D=2$.  For this section, we impose an extra constraint that each
	person should participate in exactly two tubes.  This constraint is to
	minimize the pipetting works that consume labor and suffer from
	reproducibility issues.

	First, let us walk thorough an $N=4$ example,
	after which we will give a general construction.

\subsection{A \texorpdfstring{$4$}{4}-cycle example}

	The goal of this subsection is to identify two infected persons using four
	tests on four persons, i.e., to find a $(4,4,2)$-tropical code.  Note that
	there exists a trivial $(4,4,4)$-tropical code that tests each person
	individually.  Our goal here, a $(4,4,2)$-tropical code, is not meant to be
	efficient, but to become the building block of a design based on bipartite
	graphs.

	Call the four tubes A, B, C, and D; call the four persons W, X, Y, and Z.
	Let the lower letters $a$, $b$, $c$, $d$, $w$, $x$, $y$, and $z$ be the
	corresponding Ct values.

\def\pgf@sm@shape@name{diamond}
\pgf@sh@bgpath{
	\outernortheast\pgf@xc\pgf@x\pgf@yc\pgf@y
	\pgfmathsetlength\pgf@x{\pgfkeysvalueof{/pgf/outer xsep}}
	\pgfmathsetlength\pgf@y{\pgfkeysvalueof{/pgf/outer ysep}}
	\advance\pgf@xc by-1.414213\pgf@x
	\advance\pgf@yc by-1.414213\pgf@y
	\pgfpathmoveto{\pgfqpoint{0pt}{\pgf@yc}}
	\pgfpathlineto{\pgfqpoint{-\pgf@xc}{0pt}}
	\pgfpathlineto{\pgfqpoint{-\pgf@xc}{-\pgf@yc}}
	\pgfpathlineto{\pgfqpoint{0pt}{-\pgf@yc}}
	\pgfpathlineto{\pgfqpoint{\pgf@xc}{0pt}}
}
\pgf@sh@anchorborder{
	\pgf@xa\pgf@x\pgf@ya\pgf@y
	\outernortheast
	\def\pgfnext{\pgf@marshal}
	\ifdim\pgf@ya>0pt
		\ifdim\pgf@xa<0pt
			\pgf@x-\pgf@x
		\fi
	\else
		\ifdim\pgf@xa>0pt
			\pgf@y-\pgf@y
		\else
			\def\pgfnext{
				\pgfpointborderrectangle
					{\pgfpoint{\pgf@xa}{\pgf@ya}}
					{\outernortheast}
			}
		\fi
	\fi
	\edef\pgf@marshal{
		\noexpand\pgfpointintersectionoflines
			{\noexpand\pgfpointorigin}
			{\noexpand\pgfqpoint{\the\pgf@xa}{\the\pgf@ya}}
			{\noexpand\pgfqpoint{\the\pgf@x}{0pt}}
			{\noexpand\pgfqpoint{0pt}{\the\pgf@y}}
	}
	\pgfnext
}
\tikzset{
	testnodes/.style={nodes={diamond,draw,inner sep=2,minimum size=12}},
	personedges/.style={
		every edge/.style={draw,->,bend right,shorten <=2,shorten >=2,#1},
		nodes=auto
	}
}
\begin{figure}
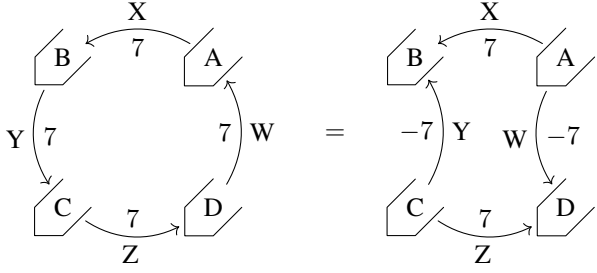

	\def\tikz@auto@anchor{\pgfmathsetmacro\tikz@anchor{atan2(-\pgf@x,\pgf@y)}}
	\def\tikz@auto@anchor@prime{\pgfmathsetmacro\tikz@anchor{atan2(\pgf@x,-\pgf@y)}}
	$$\dss
	\tikz[baseline=-axis_height]{
		\draw[testnodes]
			(-1, 1)node(B){B}(1, 1)node(A){A}
			(-1,-1)node(C){C}(1,-1)node(D){D}
		;
		\def\nodes#1#2{node{#1}node[']{$#2$}}
		\draw[personedges]
								(A)edge node[']{X}node{$7$}(B)
			(B)edge node[']{Y}node{$7$}(C)		(D)edge node[']{W}node{$7$}(A)
								(C)edge node[']{Z}node{$7$}(D);
		;
	}
	\dss=\dss
	\tikz[baseline=-axis_height]{
		\draw[testnodes]
			(-1, 1)node(B){B}(1, 1)node(A){A}
			(-1,-1)node(C){C}(1,-1)node(D){D}
		;
		\draw[personedges]
								(A)edge[]node[']{X}node{$7$}(B)
			(C)edge node[']{Y}node{$-7$}(B)		(A)edge node[']{W}node{$-7$}(D)
								(C)edge[]node[']{Z}node{$7$}(D);
		;
	}
	\dss$$
	\caption{
		Left: A directed $4$-cycle with weighted arrows to indicates the
		diff-lays; a person's (an edge's) specimen is first put in the tube
		(vertex) at the arrow tail and then added into the tube at the arrow
		head after four cycles.  Right: The direction of an arrow is part of the
		sign of the diff-lay.  A negative diff-lay can be understood as the
		opposite arrow with a positive diff-lay.  So in this case, this directed
		$K_{22}$ represents the same configuration as the directed $4$-cycle to
		the left dose.
	}\label{fig:4cycle}
\end{figure}

	Encoding:  Prepare the pools as the following.
	\[\bma{
		a	\\	b	\\	c	\\	d
	}≔\bma{
		7	&	0	&	∞	&	∞	\\
		∞	&	7	&	0	&	∞	\\
		∞	&	∞	&	7	&	0	\\
		0	&	∞	&	∞	&	7	
	}⊙\bma{
		w	\\	x	\\	y	\\	z
	}\]
	This can be graphically paraphrased by \cref{fig:4cycle} (left), in which a
	vertex is a tube and an edge is a person.  An edge's specimen is first put
	in the arrow tail and then added into the arrow head after $7$ cycles.

	Decoding: When no one is infected, all four tubes report negative results.
	When one person is infected, some two tubes will report positive results.
	For instance, if A and B are positive, then X is infected and her Ct value
	is $x=a$.  When two persons are infected, there are two possibilities.

	First possibility---neighbor patients: Two neighboring persons are
	infected iff exactly three tubes report positive results.  For instance,
	when we see that A, B and C are positive, it must be that X and Y are
	infected.  Their Ct values can be calculated by $(x,y)=(a,c-7)$.

	Second possibility---diagonal patients:  Two diagonal persons are infected
	iff all four tubes report positive results.  We cannot easily tell whether X
	and Z are sick or W and Y are sick because, as subsets, $W∪Y=X∪Z$.  That
	said, we can compute the alternating sum $a-b+c-d$.  If X and Z are
	infected, then
	\[a-b+c-d=x-(x+7)+z-(z+7)=-14.\]
	If W and Y are infected, then
	\[a-b+c-d=(w+7)-y+(y+7)-w=14.\]
	To rephrase it, the alternating sum of the Ct values along the cycle will be
	either $14$ or $-14$.  Its sign tells us if X and Z are infected or W and Y
	are infected.  When we see that X and Z are infected, their Ct values are
	$(x,z)=(a,c)$.  Otherwise, the Ct values of W and Y are $(w,y)=(d,b)$.

	The next lemma formalizes the construction.

	\begin{lem}[Nonzero $4$-cycle]\label{lem:4cycle}
		Consider a configuration like \cref{fig:4cycle} (left) where $7$'s are
		replaced with general diff-lays $δ_†W†,δ_†X†,δ_†Y†,δ_†Z†∈ℤ$.  Assume
		$δ_†W†+δ_†X†+δ_†Y†+δ_†Z†≠0$.  Then it forms a $(4,4,2)$-tropical
		code within maximum delay $\max(|δ_†W†|,|δ_†X†|,|δ_†Y†|,|δ_†Z†|)$.
	\end{lem}

	\begin{proof}
		If X and Z are infected, then
		\[a-b+c-d=x-(δ_†X†+x)+z-(δ_†Z†+z)=-δ_X-δ_†Z†;\]
		if W and Y are infected, then
		\[a-b+c-d=(δ_†W†+w)-y+(δ_†Y†+y)-w=δ_†W†+δ_†Y†.\]
		Given that $δ_†W†+δ_†Y†≠-δ_†X†-δ_†Z†$, the alternating sum $a-b+c-d$
		will be distinguishable in two cases.  Other cases (no patients, one
		patient, and neighboring patients) are straightforward.
	\end{proof}

\subsection{A \texorpdfstring{$3$}{3}-cycle non-example}

	Knowing how to deal with $4$-cycles, we now show
	that a tropical code cannot contain $3$-cycles.

	\begin{lem}[No $3$-cycle]\label{lem:3cycle}
		There does not exist a tropical code where three persons
		participate in three tests in total, each person is in
		two tests, and each test contains two persons.
	\end{lem}

	\begin{proof}
		Suppose part of the schedule matrix that
		concerns the three persons and three tests is
		\[S'=\bma{
			δ	&	ε	&	∞	\\
			ζ	&	∞	&	η	\\
			∞	&	ϑ	&	κ
		}\]
		where $0≤δ,ε,ζ,η,ϑ,κ≤ℓ$.  Then we cannot distinguish
		\[S'⊙\bma{
			2ℓ-ζ	\\	0	\\	∞
		}=\bma{
		 	ε	\\	2ℓ	\\	ϑ
		}=S'⊙\bma{
			∞	\\	0	\\	2ℓ-η
		}\]
		and hence this cannot be a tropical code.
	\end{proof}

	\Cref{lem:4cycle,lem:3cycle} imply that, in order
	to find a tropical code that tests as many people
	as possible, we should look at graphs with girth $4$.

\subsection{Bipartite Sidon graph}

\begin{figure}
	$$\dss\bma{
		7	&	0	&	0	\\
		0	&	7	&	0	\\
		0	&	0	&	7	\\
	}\dss=\dss
	\tikz[baseline=-axis_height]{
		\draw[testnodes,scale=2/3]
			foreach\x in{-2,0,2}{
				(\x,1)node(y\x){}(\x,-1)node(x\x){}
			}
		;
		\draw[personedges,delay/.code={\ifnum\x=\y\else\pgfkeysalso{-}\fi}]
			foreach\x in{-2,0,2}{
				foreach\y in{-2,0,2}{
					(y\y)edge[bend left=(\x+\y)*5,delay](x\x)
				}
			}
		;
	}
	\dss$$
	\caption{
		A diff-lay configuration of $K_{3,3}$ (which contains $K_{2,3}$).
		Arrows are $+7$.  Undirected edges are $0$ (specimen put in both
		tubes upon machine startup).  One can verify that all $4$-cycle
		sums are nonzero.
	}\label{fig:K33}
\end{figure}

\begin{figure}
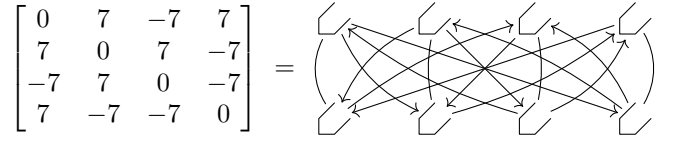

	$$\dss\bma{
		0	&	7	&	-7	&	7	\\
		7	&	0	&	7	&	-7	\\
		-7	&	7	&	0	&	-7	\\
		7	&	-7	&	-7	&	0	
	}\dss=\dss
	\tikz[baseline=-axis_height]{
		\draw[testnodes,scale=2/3]
			foreach\x in{-3,-1,1,3}{
				(\x,1)node(y\x){}(\x,-1)node(x\x){}
			}
		;
		\draw[personedges={bend left=(\x+\y)*5}]
			foreach\x/\y in{-3/-3,-1/-1,1/1,3/3}{(y\y)edge[-](x\x)}
			foreach\x/\y in{-3/-1,-1/1,-3/3}{(y\y)edge[->](x\x)(y\x)edge[->](x\y)}
			foreach\x/\y in{-3/1,-1/3,1/3}{(y\y)edge[<-](x\x)(y\x)edge[<-](x\y)}
		;
	}
	\dss$$
	\caption{
		A diff-lay configuration of $K_{4,4}$ (which contains $K_{3,4}$).
		Arrows are diff-lay $7$; undirected edges are $0$ (specimen put in both
		tubes upon machine startup).  One can verify that all $4$-cycle sums are
		nonzero.
	}\label{fig:K44}
\end{figure}

	By Mantel's theorem, the $3$-cycle-free graphs that maximize the number of
	edges are complete bipartite graphs with equal or almost equal partition
	sizes \cite{AZ18}.  These graphs have the potential to give rise to
	good tropical codes.

	Before actually constructing any tropical code out of a complete bipartite
	graph $K_{p,q}$, we demonstrate how to encode the delays.  Let $B$ be a
	$p×q$ matrix.  Consider the edge that connects $u∈[p]$ to the left to
	$v∈[q]$ to the right.  If the specimen is first put in left $u$ and, after
	$δ$ cycles, in right $v$, let $B_{uv}$ be $δ$.  If the specimen is first put
	in right $v$ and, after $ε$ cycles, in left $v$, let $B_{uv}$ be $-ε$.  This
	is how we can represent the diff-lays on a bipartite graph using a weighted
	biadjacency matrix $B$.

	As an example,
	\[B=\bma{
		7	&	-7	\\
		-7	&	7	
	}\]
	represents \cref{fig:4cycle} (right).  In a biadjacency matrix
	representation, the cycle-sum condition becomes whether every $2×2$
	sub-matrix $[\sma{δ&ε\\η&ζ}]$ satisfies $δ-ε+ζ-η≠0$.  For $K_{3,3}$,
	\cref{fig:K33} gives a valid assignment.  For $K_{4,4}$, \cref{fig:K44}
	gives a valid assignment.

	This begs the question of whether there exist a systematic way to fill-in
	larger biadjacency matrices to meet said condition.  The answer is positive.

	\begin{thm}[Two patients, min pipetting]\label{thm:pipet}
		Let $T≥2$.  Let $p≥T/2$ be an odd prime.  Then there exist a graph-based
		$(T,⌊T/2⌋⌈T/2⌉,2)$-tropical code within maximum delay $(p-1)/2$.
	\end{thm}

	\begin{proof}
		Treat the multiplication table
		\[\bma{
			1·1\bmod p	&	1·2\bmod p	&	⋯	&	1·p\bmod p	\\
			2·1\bmod p	&	2·2\bmod p	&	⋯	&	2·p\bmod p	\\
				⋮		&		⋮		&	⋰2	&		⋮		\\
			p·1\bmod p	&	p·2\bmod p	&	⋯	&	p·p\bmod p	
		}\]
		as the biadjacency matrix of a weighted directed bipartite graph.  Or
		use part of it when $p>T/2$.  To minimize the delay. use integers
		between $±(p-1)/2$ to represent the residue classes modulo $p$.  To
		verify that the cycle sum along any $4$-cycle is nonzero, note that
		every $2×2$ sub-matrix has the structure $[\sma{tb&te\\ub&ue}]$.  Hence
		the cycle sum is congruent to $tb-te+ue-ub≡(t-u)(b-e)≢0\pmod p$.
	\end{proof}

	According to Theorem~\ref{thm:pipet}, it follows that asymptotically
	there exist tropical codes such that  $T≈2√N$ and $ℓ≈√N$.
	
	We note that the tropical group testing schemes proposed in this section
	already represent an improvement over binary group testing.  A binary group
	testing scheme, when each person participates in two tubes, would avoid $3$-
	and $4$-cycles.  And it is known that the number of edges of a
	$4$-cycle--free graph is sub-quadratic \cite{17BJ}, whereas our proposed
	schemes allows $N$ to scale quadratically in terms of $T$.

\section{Two Infections and Tropical Supremacy}\label{sec:bite}

	This section again considers the $D=2$ case.  However, in this section we
	allow people to participate in more than two tests.  The goal of this
	section is to show that under this setup, there exist tropical codes that
	require less tests than the information-theoretical bound for binary group
	testing.

	We first go over an $(11,66,2)$-example.  Then we state a general sufficient
	condition of the existence of codes.   We provide a construction that
	satisfies this condition.  Finally, we briefly discuss the possibility
	of using concatenation to conserve the delays.

\subsection{An \texorpdfstring{$(11,66,2)$}{(11, 66, 2)}-example}

	For any positive integer $T$, let $[T]$ be $\{1,2,\dotsc,T\}$.
	A \emph{block design} $ℱ⊆2^{[T]}$ is a family of subsets of $[T]$.
	An element of $ℱ$ is called a \emph{block}.
	
	Let us take $T=11$ tests as an example.  Let $ℱ⊆2^{[T]}$ be a block design
	such that, for any two blocks $B,Z∈ℱ$, we have $|Z＼B|≥2$.  It can be shown
	that $|ℱ|≤66$ and the equality holds when $ℱ$ is the (unique up to
	isomorphism) $(4,5,11)$-Steiner system \cite[Section~6]{Bergstrand82}.  For
	more on Steiner systems, see the standard textbook \cite{CD06} and a recent
	breakthrough \cite{Keevash19}.  For now, fix $ℱ$ to be the
	$(4,5,11)$-Steiner system and $N=66$.

	Let $j：ℱ→[66]$ be a bijection that associates blocks with id numbers.  Let
	$k：ℱ→[37]$ be a coloring of blocks\footnote{ We found $k$ by a computer
	program that runs a greedy algorithm.  The optimality of $37$ is not the
	concern here.  We merely want an explicit number.} such that $k(B)≠k(Z)$
	whenever $|B∩Z|≥2$.  Let $x_j$ be the Ct value of the $j$th person.  Let the
	schedule be
	\[S_{tj(B)}≔\cas{
		t·k(B)\bmod37	&	if $t∈B$,	\\
		∞				&	if $t∉B$.	
	}\label{mat:mod37}\]
	Let $c_t$ be the Ct value of the $t$th tube.  Regarding the performance
	of $ℱ$ and $S$ we have the following two observations that supports the
	injectivity of tropical-multiplication of $S$ from the left.

	\begin{pro}\label{pro:single}
		For any three distinct individuals $A,B,Z∈ℱ$, we can differentiate the
		case where individuals $A,B$ are infected from where $A,Z$ are infected.
	\end{pro}

	\begin{proof}
		First notice that If $B＼A≠Z＼A$, then we can easily
		distinguish between the two cases based upon the set of 
		tests that contain infected samples.  Suppose therefore
		that $B＼A=Z＼A⊇\{t,u\}$. Then, when $A,Z$ are infected,
		\begin{align*}
			c_t-c_u
			&	=(S_{tj(Z)}+x_{j(Z)})-(S_{uj(Z)}+x_{j(Z)})	\\
			&	=S_{tj(Z)}-S_{uj(Z)}≡(t-u)k(Z)\pmod{37}
		\end{align*}
		On the other hand, when $A,B$ are infected,
		\begin{align*}
			c_t-c_u
			&	=(S_{tj(B)}+x_{j(B)})-(S_{uj(B)}+x_{j(B)})	\\
			&	=S_{tj(B)}-S_{uj(B)}≡(t-u)k(B)\pmod{37}
		\end{align*}
		As $(t-u)k(Z)≢(t-u)k(B)\pmod{37}$, this will differentiate $B$ from $Z$.
	\end{proof}

	\begin{pro}\label{pro:couple}
		For any four distinct individuals $A,B,Y,Z∈ℱ$, we can differentiate the
		case where individuals $A,B$ are infected from where $Y,Z$ are infected.
	\end{pro}

	\begin{proof}
		Suppose that $Y,Z$ are infected and we hypothesize that $A,B$ are
		infected.  The goal is to reject this hypothesis using the diff-lays.

		We say a tube $t$ is dominated by an individual $B$ if
		$c_t=S_{tj(B)}+x_{j(B)}$, i.e., the contribution of $B$ is what
		makes $c_t$ $c_t$.  Let $I_{AY}$ be the set of tubes that we
		think are dominant by $A$ but are actually dominated by $Y$.
		Define $I_{AZ}$, $I_{BY}$, and $I_{BZ}$ similarly.  Then
		$|I_{AY}∪I_{AZ}∪I_{BY}∪I_{BZ}|≥|A∪B|=|A＼B|+|B|≥2+5=7$.  This implies
		that one of $I_{AY}$, $I_{AZ}$, $I_{BY}$, and $I_{BZ}$ has cardinality
		$2$ or higher.  Suppose $|I_{AY}|≥2$ and $\{t,u\}⊆I_{AY}$, then
		\[c_t-c_u≡(t-u)k(Y)≢(t-u)k(A)\pmod{37}.\]
		Now that we can distinguish $A$ from $Y$, we can reject
		the hypothesis that $\{A,B\}$ are the infected persons.
	\end{proof}

	The moment we have rejected all incorrect hypotheses using
	\cref{pro:single,pro:couple}, whatever remains must be the
	truth.  That yields the following result.

	\begin{pro}\label{pro:mod37}
		Schedule \cref{mat:mod37} is an
		$(11,66,2)$-tropical code within maximum delay $36$.
	\end{pro}

	\begin{proof}
		We claim that the following decoder works.  Given any
		$𝐜∈(\{0\}∪ℕ∪\{∞\})^T$, if the number of positive tubes is $5$, we know
		those $5$ tubes are the only infected block.  If more than $5$ tubes are
		positive, then we know there are two patients.  We blindly guess two
		blocks $A,B∈ℱ$ and check if $\{A,B\}$ is compatible with $𝐜$.  If so,
		output the guess; if not, start over and make a new guess.
	
		Here is why the decoder, albeit inefficiently, works.  \Cref{pro:single}
		shows that if $|\{A,B\}∩\{Y,Z\}|=1$, where $Y,Z∈ℱ$ are the actual
		patients, there will be a contradiction.  \Cref{pro:couple} shows that
		if $|\{A,B\}∩\{Y,Z\}|=0$, there will also be a contradiction.  Hence the
		claimed decoder will terminate iff $|\{A,B\}∩\{Y,Z\}|=2$, i.e., our
		guess matches the reality.  There are a finite number of combinations to
		be guessed so the decoder must eventually guess correctly.  Once we know
		who are the patients it is straightforward to infer their Ct values.  We
		conclude that $S$ is a valid tropical code.
	\end{proof}

\subsection{More tests and test takers}

	One immediately sees that \cref{pro:mod37} can be
	generalized to host more tests and larger population.

	\begin{thm}[Two patients, max participants]\label{thm:modp}
		Let $ℱ⊆2^{[T]}$ be a block design with $N$ blocks.  Assume $|Z＼B|≥2$
		and $|B|≥3$ for all distinct blocks $B,Z∈ℱ$.  Let $j：ℱ→[N]$ be a
		bijection.  Let $p$ be an integer whose prime divisors are $≥T$.
		Assume there exists $k：ℱ→[p]$ satisfying $k(B)≠k(Z)$ whenever
		$|B∩Z|≥2$ (increase $p$ if necessary). Then
		\[S_{tj(B)}≔\cas{
			t·k(B)\bmod p	&	if $t∈B$,	\\
			∞				&	if $t∉B$.	
		}\]
		is a $(T,N,2)$-tropical code within maximum delay $p-1$.
	\end{thm}

	\begin{proof}
		Suppose there is one infected person, $Z∈ℱ$, and we guess there is one,
		$B∈ℱ$.  If $B≠Z$, then the tubes in $B＼Z$ will be negative while we
		expect them to be positive.  Hence we can reject this guess.

		Suppose there is one infected person, $Z∈ℱ$, but we guess there are two,
		$A,B∈ℱ$.  Without loss of generality, suppose $A≠Z$.  As $|A＼Z|≥1$,
		tubes in $A＼Z$ will be negative while we expect them to be positive.
		Hence we can reject this guess.

		Suppose there are two infected persons, $Y,Z∈ℱ$, but we guess there
		is one, $B∈ℱ$.  Without loss of generality, suppose $B≠Y$.  As
		$|Y＼B|≥1$, tubes in $Y＼B$ will be positive while we expect them to
		be negative.  Hence we can reject this guess.

		Suppose there are two infected persons, $Y,Z∈ℱ$, and we guess there
		are two, $A,B∈ℱ$.  Then depending on $|\{A,B\}∩\{Y,Z\}|=1$ or $0$, we
		will see a contradiction due to the same reasoning in \cref{pro:single}
		or \ref{pro:couple}, respectively.  Either case, we can reject incorrect
		hypothesis that $\{A,B\}$ are infected.
	\end{proof}
	
	The following lemma prepares block designs
	whose $N$ is satisfactorily large.

	\begin{lem}[{\cite[Theorem~1]{GS80}}]\label{lem:choose}
		For $T≥w≥1$, there exists a block design
		$ℱ⊆2^{[T]}$ such that each block has size $w$,
		\[|ℱ|≥÷1T\binom T{w},\]
		and $|B＼Z|≥2$ for all distinct blocks $B,Z∈ℱ$.
	\end{lem}

	Choose $w≔⌈T/2⌉$ to maximize the number of blocks.
	The asymptote of $N$ in terms of $T$ is thus
	\[N≈÷{2^T}{T√{πT/2}}.\]
	Conversely, $T≈㏒₂(N)+1.5㏒₂(㏒₂N)$.  In contrast, binary group testing has
	$2^T≥N(N-1)/2+N+1$ because the number of cases cannot exceed the number of
	test outcomes.  That implies $T≈2㏒₂(N)$.   We therefore conclude that
	tropical group testing beats the information-theoretical bound for binary
	group testing (by a factor of $1.99$).

\pgfplotstableread{
	$T$			8	9	10	11	12	13	14	15	16	
	bipartite	16	20	25	30	36	42	49	56	64	
	modulo		9	14	26	42	77	132	246	429	805	
	const		14	18	36	66	132	166	325	585	1170
	info		22	31	44	63	89	127	180	255	361	
}\tableweight
\begin{table}
	\caption{
		The number of blocks, $N$, from four sources: the bipartite graph
		construction, \cite{GS80}'s construction using modulo, the best
		known constant weight codes of minimum distance $4$ \cite{BConst},
		and the information-theoretical bounds for binary group testing.
	}\label{tab:weight}
	\def\arraystretch{1.5}
	$$\pgfplotstabletypeset[
		every head row/.style={before row=\toprule,after row=\midrule},
		every last row/.style={after row=\bottomrule},string type,
	]\tableweight$$
\end{table}

	For small parameters, a worthy reference is Brouwer's table of
	constant-weight codes of minimum distance $4$ \cite{BConst}.  See
	\cref{tab:weight} for a copy of the first few terms.  It can be seen that
	$(T,N)=(11,66)$ is the first time constant-weight codes surpass the
	information-theoretical bound, which is the reason we chose this example.

\subsection{Kronecker-Amplified Constructions}
	
	The previous subsection shows that one can easily beat the
	information-theoretical bounds using delays growing linear in $N$.  This
	subsection wants to limit the usage of delays without having to introduce
	too many additional tests.
	
	To facilitate the construction, let us practice how to assemble a code out
	of existing codes.  Heuristically speaking, our approach will allow us to
	generate new codes from previously constructed codes without increasing the
	required delay.

	Let $S$ be an $(t,n,2)$-tropical code within maximum delay $n$.  (For
	instance, schedule \cref{mat:mod37} as a $(11,66,2)$-tropical code.)
	Consider the Kronecker product
	\[S⊗𝟏_{1×n}\]
	as a $t×n²$ schedule matrix, where $𝟏_{1×n}$ is the all-one matrix of
	dimension $1×n$.  Notice that the first $n$ columns of this matrix are
	identical, the next $n$ columns of it are identical, and so on.  That is to
	say, this matrix first gathers every $n$ persons into one pool and then
	tests the resulting $n$ pools using $S$.
	
	If everyone participating the tests is healthy, then all tests will turn
	out negative.  Suppose one person, her id being $0≤Z≤n²-1$, is infected.
	Then the tests will reveal $⌊Z/n⌋$, the tens (the second least significant)
	digit of $Z$ in the $n$-ary expression, and $z_*$, the Ct value of $Z$.  If
	two persons, $0≤Y,Z≤n²-1$, are infected, then the tests will reveal
	$\{⌊Y/n⌋,⌊Z/n⌋\}$, i.e., the tens digits of $Y$ and $Z$.  There are two
	cases.
	\begin{itemize}
		\item	If the digit differ, $⌊Y/n⌋≠⌊Z/n⌋$, then $S⊗𝟏_{1×n}$
				sees two infected pools and will report $y_*$ and $z_*$.
		\item	If the digit agree, $⌊Y/n⌋=⌊Z/n⌋$, then $S⊗𝟏_{1×n}$
				sees one infected pool and will report $\min(y_*,z_*)$.
	\end{itemize}

	A similar argument shows that $𝟏_{1×n}⊗S$ reveals $Y\bmod n$ and
	$Z\bmod n$, the ones digits of $Y$ and $Z$, together with one (if the
	digits agree) or two (if the digits differ) Ct values.
	
	We now investigate what happens when we combine $S⊗𝟏_{1×n}$ and
	$𝟏_{1×n}⊗S$.  Define $S^{(2)}$ as this $2t×n²$ schedule matrix
	\[S^{(2)}≔\bma{
		𝟏_{1×n}⊗S	\\
		S⊗𝟏_{1×n}
	}\]
	and use it to perform the tests.  If there are zero infected or if there is
	one, the decoding procedure is trivial.  Therefore, assume that there are
	two infected persons, $0≤Y,Z≤n²-1$.  Let their Ct values be $y_*$ and $z_*$.
	Let $y_{10}n+y₁$ and $z_{10}n+z₁$, where $0≤y₁,y_{10},z₁,z_{10}≤n-1$, be
	the $n$-ary expansions of $Y$ and $Z$, respectively.  Then the first $t$
	tests can tell us $\{y₁,z₁\}⊆\{0,\dotsc,n-1\}$.  More precisely, the
	first $t$ tests tell us $\{(y₁,y_*),(z₁,z_*)\}$ when $y₁≠z₁$ and tell us
	$\{(z₁,\min(y_*,z_*))\}$ when $y₁=z₁$.  Similarly, the last $t$ tests can
	tell us $\{(y_{10},y_*),(z_{10},z_*)\}$ when $y_{10}≠z_{10}$ and tell us
	$\{(z_{10},\min(y_*,z_*))\}$ when $y_{10}=z_{10}$.  Now there is only one
	bit of information left consider. In particular, we need to differentiate
	between the case where
	\[\{Y,Z\}=\{y_{10}n+y₁,z_{10}n+z₁\}\phantom.\]
	and the case where
	\[\{Y,Z\}=\{y_{10}n+z₁,z_{10}n+y₁\}.\]
	
	When $y₁=z₁$,  we already can infer $\{Y,Z\}=\{z₁+y_{10}n,\allowbreak
	z₁+z_{10}n\}$.  When $y_{10}=z_{10}$, similarly, we already can infer
	$\{Y,Z\}=\{y₁+z_{10}n,z₁+z_{10}n\}$.  If $|\{y₁,z₁\}|=|\{y_{10},z_{10}\}|=2$
	and $Y$ and $Z$ have different Ct values, we will see that $y₁$ and $y_{10}$
	associate to a Ct value---$y_*$---whilst $z₁$ and $z_{10}$ associate to
	another different Ct value---$z_*$.  That will help us link $y₁$ to $y_{10}$
	and $z₁$ to $z_{10}$ and thus help us recover $Y,Z$.  One difficult case
	remains to be addressed, namely $|\{y₁,z₁\}|=|\{y_{10},z_{10}\}|=2$ yet
	$y_*=z_*$.
	
	To overcome the last case, we add one more test to make up the missing
	information: Let $p≥\max(3,n)$ be the smallest possible prime.  Let
	$0≤α₁,α₂,β₁≤p-1$ be distinct numbers modulo $p$.  Compute, for any
	$0≤a₁,a₂≤n-1$,
	\[\bma{
		b₁
	}≔\bma{
		1	&	β₁
	}\bma{	
		1	&	α₁	\\
		1	&	α₂
	}^{-1}\bma{
		a₁	\\	a₂
	}\pmod p\]
	using modulo $p$ arithmetics.  That is, we find $b₁$ such that $(α₁,a₁)$
	and $(α₂,a₂)$ and $(β₁,b₁)$ are point--evaluation pairs of some degree-one
	polynomial.  In other words, they are colinear in the plane $𝔽_p²$.

	Treating $b₁$ as a function in $a₁,a₂$, we want to append $b₁(a₁,a₂)$ at
	the bottom of the $(1+a₁+a₂n)$th column of $S^{(2)}$.  That is, we stack
	$S^{(2)}$ on top of this $1×n^k$ matrix
	\[Q^{(2)}≔\bma{
		1	&	β₁
	}\bma{	
		1	&	α₁	\\
		1	&	α₂
	}^{-1}\bma{
		𝟏_{1×n}⊗M	\\
		M⊗𝟏_{1×n}		
	}\pmod p\]
	where
	\[M≔\bma{
			0	&	1	&	⋯	&	n-1
	}.\]
	We claim that the addition of $Q^{(2)}$ to $S^{(2)}$ results in a valid
	tropical code.

	\begin{pro}\label{pro:code2}
		If $S$ is an $(t,n,2)$-tropical code within maximum delay $n$, then
		\[\bma{S^{(2)}\\Q^{(2)}}\]
		is an $(2t+1,n²,2)$-tropical code within maximum delay $p-1$, where
		$p$ is the least prime $≥\max(3,n)$.
	\end{pro}

	\begin{proof}
		When there is $≤1$ infected person, there are two infected persons
		sharing a same digit, or there are two infected persons having different
		Ct values, the first $2t$ tests suffice.  When the two infected persons
		have distinct digits yet the same Ct values, we utilize the last test in
		the following manner.

		Let $Y$ and $Z$ be the infected persons, $0≤Y,Z≤n²-1$.  Let
		$y_{10}n+y₁$ and $z_{10}n+z₁$ be the $n$-ary expansion of
		$Y$ and $Z$.  The result of $Q^{(2)}$, denoted by $c₁$, is
		$c₁=\min(b₁(y₁,y_{10})+y_*,b₁(z₁,z_{10})+z_*)$.  Therefore
		$c₁-z_*∈\{b₁(y₁,y_{10}),b₁(z₁,z_{10})\}$.  Now there are five
		points on the $𝔽_p²$ plane:
		\begin{align*}
			(α₁,y₁)\,,\,(α₂,y_{10})\,,	&						\\*[-1.5ex]
										&	\;(β₁,c₁-z_*)\,,	\\*[-1.5ex]
			(α₁,z₁)\,,\,(α₂,z_{10})\,.	&							
		\end{align*}
		We know the point $(β₁,c₁-z_*)$ is either $(β₁,b₁(y₁,y_{10}))$ or
		$(β,b₁(z₁,z_{10}))$.  It must be colinear with two of the four
		points to its left.  By looking at which two are colinear with it,
		we can correctly link $y₁$ to $y_{10}$ and $z₁$ to $z_{10}$.
		This recovers $Y$ and $Z$ and finishes the proof.
	\end{proof}
	
	Next we define the notion of BITE.

	\begin{dfn}
		A $(T,N,2)$-tropical code is said to \emph{beat the
		information-theoretical estimate (BITE)}\footnote{
		We will use BITE as a verb.} if $N^2/4>2^T$.
	\end{dfn}

	BITing is slightly stronger than simply beating the binary
	bound (by about one test), but it is an inductive invariant.

	\begin{pro}\label{pro:bite2}
		If $S$ is an $(t,n,2)$-tropical code that BITEs, then
		\[\bma{
			S^{(2)}	\\
			Q^{(2)}
		}\]
		is an $(n²,2t+1,2)$-tropical that BITEs.
	\end{pro}

	\begin{proof}
		The precondition on $S$ is $n²/4>2^t$.  What is desired to be
		shown is $(n²)²/4>2^{2t+1}$.  The former implies the latter.
	\end{proof}

	\Cref{pro:code2,pro:bite2} generalize to the following
	theorem, of which the proof is given in \cref{pf:bite}.

	\begin{thm}[Two patients, no patience]\label{thm:bite}
		Let $S$ be a $(t,n,2)$-tropical code within maximum delay $n$,
		then there is an $(n^k,kt+2k-3,2)$-tropical code within
		$q-1$ cycles, where $q≥\max(3k-3,n)$ is a prime power.
		If $S$ BITEs, so do the putative codes BITE for all $k≥1$.
	\end{thm}
	
	\begin{rem}
		This Kronecker-based construction is inspired by
		\cite[Theorem~1]{Edel04} (\cite{Mukhopadhyay78}) and the way it is
		used in the paper.  This shares common elements with the grid-based
		construction \cite{BK20} that is used to attack the pandemic
		\cite{SKH20, Taufer20, MNBSURNMRNNSMNNUMMMNTN21}.  This also shares
		common elements with the fast decoder approach \cite{NPR11}. Other fast
		decoder approach such as \cite{LCPR19, BCSYZ21} also contain similar
		ideas.
	\end{rem}

	Ideally we can use very large but almost equal $q≈n≈3k-3$ by
	consulting \cref{thm:bite,lem:choose}.  Hence,  as $q$ goes to
	infinity, there exist $(T,N,2)$-tropical codes within maximum
	delay $q$, where $q≈3㏒₂(N)/㏒₂(3㏒₂(N))$ and $T≈1.01㏒₂(N)$.

\section{Many Infections and Disjunction}\label{sec:disjunct}

	In this section, we move on to general $D$.
	Concerning the existence of tropical codes, we will prove
	one necessary condition and two sufficient conditions.

	Given a schedule matrix $S∈(\{0\}∪ℕ∪\{∞\})^{T×N}$, the \emph{underlying
	block design} registers the tests each person participates in.
	It is defined to be a multiset as follows.
	\[ℱ≔\Bigl\{\{t∈[T]｜S_{tj}<∞\}\Bigm|j∈[N]\Bigr\}.\]
	A nasty edge case is when two individuals participate in exactly the same
	subset of tests.  By letting $ℱ$ be a multiset we have $|ℱ|=N$.
	Also by distinct blocks $B₁,\dotsc,B_D∈ℱ$, we mean that the $B$'s originate
	from distinct individuals. But as subsets of $[T]$ they are not necessarily
	distinct.
	
	Here is the necessary condition promised.  We state the weak version
	following by the strong version.
	
	\begin{dfn}
		A block design $ℱ$ is said to be \emph{$D$-disjunct}
		if $|Z＼(B₁∪\dotsb∪B_D)|≥1$ for distinct blocks $Z,B₁,\dotsc,B_D∈ℱ$.
	\end{dfn}

	\begin{thm}
		The underlying block design $ℱ$ of a $(T,N,D)$-tropical
		code must be $(D-1)$-disjunct.
	\end{thm}

	\begin{proof}
		Suppose $ℱ$ is not $(D-1)$-disjunct, then there exist distinct blocks
		$Z,B₁,\dotsc,B_{D-1}∈ℱ$ such that $Z⊆B₁∪\dotsb∪B_{D-1}$.  This means
		that, when the $B$'s are severely infected, we cannot tell if $Z$ is
		slightly infected or not infected.
	\end{proof}

	\begin{dfn}
		A block design $ℱ$ is said to be \emph{$D$-uniquely-disjunct}
		if it is $D$-disjunct and
 		\[\biggl|\Bigl\{Z∈ℱ\Bigm|Z＼(B₁∪\dotsb∪B_D)=\{t\}\Bigr\}\biggr|≤1\]
		for any vertex $t∈[T]$ and distinct blocks $B₁,\dotsc,B_D∈ℱ$.
	\end{dfn}

	\begin{thm}
		The underlying block design $ℱ$ of a $(T,N,D)$-tropical
		code must be $(D-1)$-uniquely-disjunct.
	\end{thm}

	\begin{proof}
		We already see that $ℱ$ must be $(D-1)$-disjunct.
		Suppose it is not $(D-1)$-uniquely-disjunct, then there exist
		distinct blocks $Y,Z,B₁,\dotsc,B_{D-1}∈ℱ$ such that
		$Y＼(B₁∪\dotsb∪B_{D-1})=Z＼(B₁∪\dotsb∪B_{D-1})$, which contains but one
		vertex.  This means that, when the $B$'s are heavily infected, we cannot
		tell apart if it is $Y$ or $Z$ that is infected.
	\end{proof}

	Here are the two sufficient conditions promised.

	\begin{dfn}
		A block design $ℱ$ is said to be \emph{$¯D¯$-separable}
		if $Z₁∪\dotsb∪Z_D≠B₁∪\dotsb∪B_D$ for different subsets of
		up to $D$ blocks, $ℱ⊇\{Z₁,\dotsc,Z_D\}≠\{B₁,\dotsc,B_D\}⊆ℱ$.
	\end{dfn}

	\begin{thm}
		If there is a $¯D¯$-separable block design $ℱ$ with $N$ blocks
		on $T$ vertices, then there exists a $(T,N,D)$-tropical code
		within maximum delay $0$.
	\end{thm}

	\begin{proof}
		Use the vanilla schedule
		\[S_{tj(B)}≔\cas{
			0	&	if $t∈B$,	\\
			∞	&	if $t∉B$,	
		}\]
		where $j：ℱ→[N]$ is a bijection.  To decode, reinterpret every test
		result as positive or negative.  Use a binary group testing decoder to
		determine who are infected.  For every infected individual $Z$, the
		“outcrop” $Z＼(B₁∪\dotsb∪B_{D-1})$ is never empty because
		$¯D¯$-separability implies $(D-1)$-disjunction \cite{CH07}.  This
		further implies that there is at least one tube wherein $Z$ is the only
		infected participant.  The Ct value of this tube is the Ct value of $Z$.
	\end{proof}

	\begin{dfn}
		A block design $ℱ$ is said to be \emph{$D$-doubly-disjunct}
		if $|Z＼(B₁∪\dotsb∪B_D)|≥2$ for distinct blocks $Z,B₁,\dotsc,B_D∈ℱ$.
	\end{dfn}

	\begin{thm}\label{thm:cycle}
		If there is a $(D-1)$-doubly-disjunct block design $ℱ$ with $N$ blocks
		on $T$ vertices, then there exists a $(T,N,D)$-tropical code.
	\end{thm}

	\begin{proof}
		定—{\mathrel{-\mkern-4mu-}}
		Use schedule
		\[S_{tj(B)}≔\cas{
			2^{t+j(B)T}	&	if $t∈B$,	\\
			∞			&	if $t∉B$,	
		}\]
		where $j$ is a bijection $j：ℱ→[N]$.

		Let $B₁,\dotsc,B_D$ be the blocks that we think are infected.  Let
		$Z₁,\dotsc,Z_D$ be the blocks that are actually infected.  For now,
		assume that they are all distinct, as otherwise the theorem statement
		will be similar (perhaps easier) to prove.  Let $I_{uv}$ be the subset
		of tests that we think are dominated by $B_u$ but actually are dominated
		by $Z_v$, where dominance means $c_t=S_{tj(B)}+x_{j(B)}$.  By that $ℱ$
		is $(D-1)$-doubly-disjunct,
		\[\Bigl|⋃_{uv}I_{uv}\Bigr|=\Bigl|⋃_uB_u\Bigr|≥2D.\]

		Next, define a bipartite graph $G$ with left part $[D]$ and right part
		$[D]$: for any $(u,v)∈[D]×[D]$, connect $u$ to the left to $v$ to the
		right $|I_{uv}|$ times.  Now $G$ has $2D$ vertices and $≥2D$ edges,
		hence it contains a cycle (possibly a $2$-cycle).  Let this $2Ψ$-cycle
		be
		
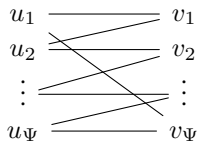
\[\begin{tikzcd}[row sep=0em,column sep=4em,every arrow/.style={draw,-}]
			u₁	&	v₁	\lar\dlar	\\
			u₂	&	v₂	\lar\dlar	\\
			⋮	&	⋮	\lar\dlar	\\
			u_Ψ	&	v_Ψ	\lar\ar{uuul}
		\end{tikzcd}\]
		\tikzset{
			every picture/.style=
				{baseline=-axis_height,shorten <=.3em,,shorten >=.3em},
			nodes={inner sep=0,font=\strut},l/.style=left,r/.style=right
		}%
		where the $u$'s are to the left and the $v$'s are to the right.
		Identify $u_{Ψ+1}≔u₁$. Every edge involved, be it
		\tikz\draw(0,0)node[l]{$u_ψ$}--(2em,0)node[r]{$v_ψ$}; or
		\tikz\draw(0,-.5ex)node[l]{$u_{ψ+1}$}--(2em,.5ex)node[r]{$v_ψ$};,
		corresponds to a test in $I_{u_ψv_ψ}$ or in $I_{u_{ψ+1}v_ψ}$,
		respectively.  Hence we can talk about the their Ct values, denoted by
		$c_ψ$ or $d_ψ$, respectively.

		Finally, examine the alternating sum along this cycle
		\[∑_{ψ=1}^{Ψ}c_ψ-d_ψ.\]
		In our mind, we expect that this is a sum of some of $B$'s diff-lays.
		But in reality, this is a sum of some of $Z$'s diff-lays.  Since the
		delays are distinct powers of two, there is no way a sum of delays
		equal to another sum of delays.  This shows that we can always reject
		incorrect guesses.
	\end{proof}

	\begin{rem}
		We understand that using exponential delays is rhetorical as the delays
		already live in the logarithmic realm.  This can be avoided.  All we
		need is that the ``cycle sums'' do not vanish.  And so it suffices to
		use random delays and bound from above the probability they vanish.
		Since only short cycles are those that are likely to vanish and since
		there are only polynomially many short cycles, we expect that a
		$(T,N,D)$-tropical code exists within polynomial delay.
	\end{rem}

	For works that discuss disjunct-ness vs other properties, see
	\cite{CH07, FFGMS21}.  For how to design block systems with disjunct and/or
	separable properties, see Kautz--Singleton \cite{KS64} and the follow-ups
	\cite{DR82, PR11, IKWO19, BPS21}.
	
	In the next two sections, we turn our interest to adaptive strategies.

\section{Two Infections and Adaptive Strategies}\label{sec:pad}

	Recall that the $D=1$ case was optimally solved in \cref{sec:diff-lay}.
	Recall also that \cref{sec:pipet,sec:bite} discussed some nonadaptive
	strategies of the $D=2$ case.  In this section, we explore adaptive
	strategies for the $D=2$ case, We will give a two-round, four-test strategy
	that finds two infected persons among arbitrarily many.  That is, we will
	construct $2$-$(4,N,2)$-tropical protocols for all $N$.  It is rather
	surprising that, compared to \cref{sec:pipet,sec:bite}, allowing a second
	round saves such a large number of tests.

\subsection{Five tests}

	Let a population have Ct values
	\[𝐱≔\bma{
		x₁	\\	⋮	\\	x_N
	}\]
	where $N$ is the number of persons being tested.
	Begin with two tests that imitate \cref{mat:hinge}.
	\[\bma{
		a	\\	b
	}≔\bma{
		1	&	2	&	⋯	&	N-1	&	N	\\
		N	&	N-1	&	⋯	&	2	&	1	
	}⊙𝐱\label{mat:5Tab}\]
	If tubes $a$ and $b$ are negative, no one is infected and we are done.
	Assume the opposite, that $a$ and $b$ are positive.  We compute the pointer
	$j≔(a-b+N+1)/2$ and double-check the $j$th person
	\[\bma{
		c
	}≔\bma{
		∞	&	＾{j-1}{⋯}	&	∞	&	0	&	∞	&	＾{N-j}{⋯}	&	∞
	}⊙𝐱,\label{mat:5Tc}\]
	here $0$ is at the $j$th column.  If $c$ is positive, we test
	\[\bma{
		d^+	\\	e^+
		}≔\bma{
			0	&	⋯	&	0	&	∞	&	1	&	⋯	&	N-j	\\
			j-1	&	⋯	&	1	&	∞	&	0	&	⋯	&	0	
		}⊙𝐱.\label{mat:5T+de}\]
	If $c$ is negative, we test
	\[\bma{
		d^-	\\	e^-
	}≔\bma{
		0	&	＾{j-1}{⋯}	&	0	&	∞	&	∞	&	＾{N-j}{⋯}	&	∞	\\
		∞	&			⋯	&	∞	&	∞	&	0	&			⋯	&	0	
		}⊙𝐱.\label{mat:5T-de}\]
	We claim the following.
	
	\begin{pro}\label{pro:2D5T}
		Schedule \crefrange{mat:5Tab}{mat:5T-de}
		form a $3$-$(5,N,2)$-tropical protocol.
	\end{pro}

	Note that schedule \cref{mat:5T+de} is just \eqref{mat:hinge} without the
	$j$th person so it can find us the second patient given that $j$ is the
	first.  It remains to show why schedule \cref{mat:5T-de} can find us two
	infected persons so quickly given that $j$ is not one.  A lemma is placed
	here, before the proof of the proposition, to help clarify what can we learn
	from the test results $a$ and $b$.

	\begin{lem}\label{lem:split}
		Given $a$ and $b$ as defined with schedule
		\cref{mat:5Tab} and $j≔(a-b+N+1)/2$.  Then
		$\min_{i≤j}x_i+i=a$ and $\min_{k≥j}x_k+(N+1-k)=b$.
	\end{lem}
	
	This lemma encodes the core idea of this section:  After $a$ and $b$, we
	learn that $j$ is likely to be infected.  That is why we query $c$ to check.
	And even if she comes out healthy, we still know that someone $i<j$ to the
	left dominates $a$ and someone $k>j$ to the right dominates $b$.  Now $i$
	needs one more test to locate; and $k$ needs another test to locate.  Those
	are what $d^-$ and $e^-$ do, respectively.

	\begin{proof}[Proof of \cref{lem:split}] By the configuration of the
		second test, $b≤x_k+(N+1-k)$, which implies $x_k≥b-(N+1-k)$ for all $k$.
		Suppose $i$ is the index that attain the minimum $a≔\min_ix_i+i$.  We
		have $0=a-x_i-i≤a-b+(N+1-i)-i=a-b+N+1-2i=2j-2i$.  This forces $2i≤2j$.
		As the index that attain the minimum must be $≤j$, we might as well
		restrict the domain of the minimum and write $a=\min_{i≤j}x_i+i$.  The
		second statement of the lemma holds by symmetry.
	\end{proof}

	\begin{proof}[Proof of \cref{pro:2D5T}]
		Let $c$ be negative and let $d^-,e^-$ be defined with schedule
		\cref{mat:5T-de}.  From the configuration of $d^-$ we know
		$\min_{i<j}x_i=d^-$.  By \cref{lem:split} we know $\min_{i≤j}x_i+i=a$.
		We can replace the domain $i≤j$ with $i<j$ because $j$ is confirmed to
		be healthy.  So far we have collected the information
		\[\bma{
			1	&	2	&	⋯	&	j-1	\\
			0	&	0	&	⋯	&	0	
		}⊙\bma{
			x₁	\\	⋮	\\	x_{j-1}
		}=\bma{
			a	\\	d^-
		}\]
		and we can infer that one infected person is $i≔a-d^-$
		with Ct value $x_i=d^-$.  By symmetry, the other infected
		person is $k≔N+1-(b-e^-)$ with Ct value $x_k=e^-$.
	\end{proof}

\subsection{Four tests}

	As it turns out, we can superimpose schedule \cref{mat:5T+de,mat:5T-de}
	in a judiciously way to optimize test $c$ away.  Let
	\[\arraycolsep2.5pt
	\bma{
		d	\\	e
	}≔\bma{
		0		&	⋯	&	0	&	∞	&	γ+1	&	⋯	&	γ+N-j	\\
		γ+j-1	&	⋯	&	γ+1	&	∞	&	0	&	⋯	&	0	
	}⊙𝐱.\label{mat:4Tde}\]
	where $γ$ is a gargantuan number, e.g., $888(N+a+b)$.

\begin{figure}
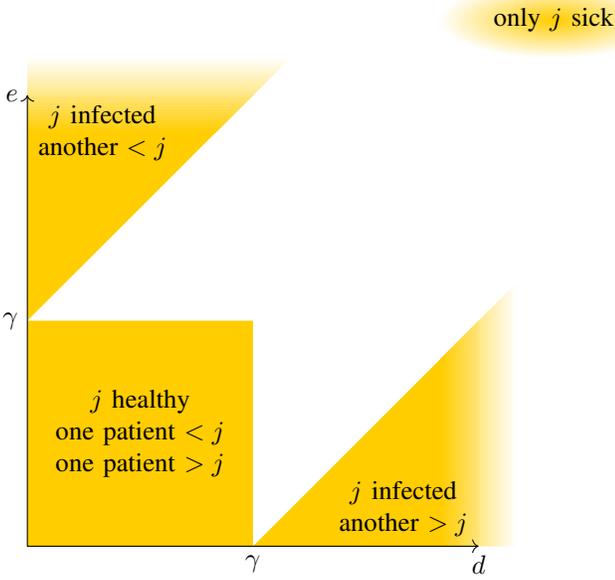

	$$\tikz{
		\begin{scope}
			\clip(0,3)|-(3.5,6.5)(0,0)rectangle(3,3)(3,0)-|(6.5,3.5);
			\fill[PMS116](0,0)rectangle(6,6);
			\shade[top color=white,bottom color=PMS116](0,5.5)rectangle(4,6.5);
			\shade[left color=PMS116,right color=white](5.5,0)rectangle(6.5,4);
		\end{scope}
		\fill[PMS116,nodes={text=black,align=center}]
			(1,5.5)node{$j$ infected\\another $<j$}
			(0,3)node[left]{$γ$}
			(1.5,1.5)node{$j$ healthy\\one patient $<j$\\one patient $>j$}
			(3,0)node[below]{$γ$}
			(5,.5)node{$j$ infected\\another $>j$}
		;
		\shade[inner color=PMS116,outer color=white]
			(7,7)circle[x radius=1.5,y radius=.5]node{only $j$ sick}
		;
		\draw[<->](0,6)node[left]{$e$}|-(6,0)node[below]{$d$};
	}\hskip0pt minus1in$$
	\caption{
		The configurations space of tests $d$ and $e$ defined with schedule
		\cref{mat:4Tde}.  The cloud at the upper right corner is at $(∞,∞)$.
	}\label{fig:wings}
\end{figure}

	\begin{thm}\label{thm:pad}
		Schedule \cref{mat:5Tab,mat:4Tde}
		give a $2$-$(4,N,2)$-tropical protocol.
	\end{thm}

	\begin{proof}
		It suffices to show that we can distinguish the following four cases:
		\begin{itemize}
			\item	$j$ is indeed infected; everyone else is healthy.
			\item	$j$ is indeed infected; the other patient is $<j$.
			\item	$j$ is indeed infected; the other patient is $>j$.
			\item	$j$ is healthy; one patient is $<j$ and the other $>j$.
		\end{itemize}
		They one-to-one correspond to the following regions of 
		the configuration space of the test results (cf.\ \cref{fig:wings}):
		\begin{itemize}
			\item	$d$ and $e$ are negative.
			\item	$d$ and $e$ are positive and $e>γ+d≥γ$.
			\item	$d$ and $e$ are positive and $d>γ+e≥γ$.
			\item	$d$ and $e$ are positive and $<γ$.
		\end{itemize}
		The four regions are mutually disjoint and exhaustive so we can and only
		have to distinguish them.  Afterward, the following are how we compute
		who are the patients and their Ct values for each of the four cases:
		\begin{itemize}
			\item	$x_j=(a+b-N-1)/2$.
			\item	$(x_{j-(e-γ-d)},x_j)=(d,(a+b-N-1)/2)$.
			\item	$(x_j,x_{j+(d-γ-e)})=((a+b-N-1)/2,e)$.
			\item	$(x_{a-d},x_{N+1-(b-e)})=(d,e)$.
		\end{itemize}
		It is straightforward to verify these deviations.
	\end{proof}

	Next, we proceed to how to identify more infected people.

\section{More Infections and Adaptive Strategies}\label{sec:nim}

	Our next goal is to give a construction that finds all $D$ infected persons
	in $3D+1$ tests regardless of how large $D$ and $N$ are.  A pilot
	construction using $7D+1$ tests is specified in the next subsection.

\subsection{A search-and-verify protocol}

	Suppose that there are $D$ persons infected among a population of $N$.
	Suppose also that we lack the knowledge of $D$.

	To begin, query
	\[\bma{
		a
	}≔\bma{
		1	&	2	&	⋯	&	N
	}⊙𝐱.\label{mat:7DTa}\]
	If $a$ is negative, we conclude that everyone is healthy.  We had used
	$1≤7·0+1$ tests and that is within the budget.  This is the \emph{leaf}
	of our recursive algorithm whose definition continues below.

	If $a$ is positive, query
	\[\bma{
		b
	}≔\bma{
		N	&	N-1	&	⋯	&	1
	}⊙𝐱.\label{mat:7DTb}\]
	Compute the index $j≔(a-b+N+1)/2$.  Then query
	\[\bma{c}≔\bma{
		∞	&	＾{j-1}{⋯}	&	∞	&	0	&	∞	&	＾{N-j}{⋯}	&	∞	\\
	}⊙𝐱\label{mat:7DTc}\]
	where the delay $0$ is at the $j$th column.  Test $c$ tells us whether $j$
	is really sick or not.  If $c$ ends up positive, $j$ is infected with Ct
	value $x_j=c$.  We then remove her from the population.  Now that there are
	$D-1$ patients left in the remaining $N-1$ persons, we start over.  Since
	the number of infected persons decreases, we expect that this recursive
	algorithm will eventually reach the leafs and return.

	If test $c$ ends up negative, we have a strengthened
	\cref{lem:split} that applies to an unknown $D$.

	\begin{lem}\label{lem:split'}
		Let $a$ and $b$ be defined with schedule \cref{mat:7DTa,mat:7DTb}
		and $j≔(a-b+N+1)/2$.  Then $\min_{i≤j}x_i+i=a$ and
		$\min_{k≥j}x_k+(N+1-k)=b$.  (Proof is identical to that of
		\cref{lem:split} but this time $𝐱$ can contain more than two patients.)
	\end{lem}

	Being told that $j$ is healthy, we know that the first $j-1$ persons contain
	at least one infected person and so do the last $N-j$ persons. We hereby
	split the population into two halves---$[1,j-1]$ and $[j+1,N]$---and apply
	the same algorithm to them separately. Since both halves contain less than
	$D$ infected persons, the recursive algorithm will eventually reach the
	leafs and return.  The only problem is, How many tests does it consume
	before termination?

\subsection{Test number analysis}

	For every three tests we spend on the searching
	\cref{mat:7DTa,mat:7DTb} and the confirmation
	\cref{mat:7DTc}, either of the following happens.
	\begin{itemize}
		\item	$a$ and $b$ indicate that $j$ is suspicious
				and $c$ confirms that she is indeed infected.
		\item	$a$ and $b$ indicate that $j$ is suspicious but
				$c$ shows that $j$ is healthy yet we can split the population
				into two halves, each containing fewer infected persons.
	\end{itemize}
	There are $D$ infected persons so the number of tests spent on the first
	case, searching and confirming, is exactly $3D$.  There are $D-1$ “gaps” we
	can split the population at so the number of tests spent on the second case,
	splitting, is at most $3D-3$.  In total, the cost is at most $6D-3$ tests.
	On top of that, every time we confirm an infected person $j$ in some
	interval $[i,k]$, the protocol will then query
	\[\bma{
		1	&	2	&	⋯	&	k-i
	}⊙\bma{
		x_i	\\	⋮	\\	x_{j-1}	\\	x_{j+1}	\\	⋮	\\	x_k
	}\]
	and sometimes the test result says that everyone involved
	is healthy.  This brings the total to $6D-3+D=7D-3$ tests.
	The final number is $7D-3$ and it is $≤7D+1$.

	\begin{pro}
		Schedule \cref{mat:7DTa,mat:7DTb,mat:7DTc}
		constitute a $(7D+1)$-$(7D+1,N,D)$-tropical protocol.
	\end{pro}

	Much to our surprise, this protocol does not depend on how many people are
	being tested.  Moreover, this protocol tells us the number of patients as a
	part of the output---we do not have to know and tell the protocol the
	number $D$ before the protocol begins.

\subsection{Upgrade the third test}

	As it turns out, some tests in the $(7D+1)$-protocol are redundant.  The key
	is that the confirmation \cref{mat:7DTc} can be combined with the searching
	\cref{mat:7DTa} of the first $j-1$ persons in the following way
	\[\bma{
		c^♯
	}≔\bma{
		j-1	&	⋯	&	1	&	0	&	∞	&	＾{N-j}{⋯}	&	∞
	}⊙𝐱.\label{mat:3DTc}\]
	This test behaves like a verification of whether $j$ is really
	sick.  At the same time it prefetches the result of the searching
	matrix in the next level of the recursion should the verification fail.

	\begin{lem}\label{lem:monotony}
		Let $a$, $b$, and $c^♯$ be defined with searching
		\cref{mat:7DTa,mat:7DTb} and multitask \cref{mat:3DTc}.
		Let $j≔(a-b+N+1)/2$.  We have
		\begin{itemize}
			\item	$x_j+j≥a=\min_{i≤j}x_i+i$,
			\item	$x_j+(N+1-j)≥b=\min_{k≥j}x_k+(N+1-k)$,
			\item	$x_j≥c^♯≔\min_{i≤j}x_i+(j-i)$, and
			\item	$c^♯≥(a+b-N-1)/2$.
		\end{itemize}
	\end{lem}

	\begin{proof}
		The first two statements are by \cref{lem:split'} and the
		third statement by the definition of $c^♯$.  Only the last one
		is nontrivial so let us prove it.  Suppose $i$ is the index
		that attains the minimum: $c^♯:=\min_{i≤j}x_i+(j-i)$.  Then
		$a+b-N-1≤x_i+i+x_i+(N+1-i)-N-1=2x_i≤2x_i+(j-i)=2c^♯$.
		This finishes the proof.
	\end{proof}

	These four inequalities enjoy a dichotomous behavior.
	
	\begin{lem}\label{lem:dichotomy}
		Let $a$ and $b$ and $c^♯$ be defined with searching
		\cref{mat:7DTa,mat:7DTb} and verify--prefetching \cref{mat:3DTc}.
		Let $j≔(a-b+N+1)/2$.  Either the following four hold
		\begin{enumerate}
			\item[(i)]	$x_j+j=a$,
			\item[(ii)]	$x_j+(N+1-j)=b$,
			\item[(iii)]$x_j=(a+b-N-1)/2$,
			\item[(iv)]	$c^♯=(a+b-N-1)/2$,
		\end{enumerate}
		or the following four hold
		\begin{itemize}
			\item	$x_j+j>a$ (thus $a=\min_{i<j}x_i+i$),
			\item	$x_j+(N+1-j)>b$ (thus $b=\min_{k>j}x_k+(N+1-k)$),
			\item	$x_j>(a+b-N-1)/2$,
			\item	$c^♯>(a+b-N-1)/2$.
		\end{itemize}
	\end{lem}

	\begin{proof}
		Due to \cref{lem:monotony}, it suffices to prove that (i), (ii),
		(iii), and (iv) are equivalent to each other.  (i) is equivalent
		to $x_j=a-j=a-(a-b+N+1)/2=(a+b-N-1)/2$.  (ii) is equivalent to
		$x_j=b-(N+1-j)=b-(N+1-(a-b+N+1)/2)=(a+b-N-1)/2$.  Hence (i), (ii),
		and (iii) are equivalent to each other.  If they hold, then
		due to the inequalities $x_j=(a+b-N-1)/2≤c^♯≤x_j$, (iv) holds.
		Conversely, suppose that (iv) holds.  Let $i≤j$ be the index that
		attains the minimum: $c^♯≔\min_{j≤j}x_i+(j-i)$.  Then the inequalities
		$2c^♯=a+b-N-1≤x_i+i+x_i+(N+1-i)-N-1≤2x_i≤2x_i+2(j-i)=2c^♯$ squeeze.
		Hence $2(j-i)=0$ and (i), (ii), and (iii) hold.  This finishes the proof.
	\end{proof}

	Thanks to the dichotomy related to $c^♯$, we either confirm that $j$
	is indeed infected when we see (iv) holds (plus we know her Ct value
	$x_j=(a+b-N-1)/2$) or we know we can split the population at $j$ when we
	see (iv) does not hold.  This leads to a strategy that only queries $3D+1$
	times.

\subsection{\texorpdfstring{$3D+1$}{3D + 1} adaptive tests
	diagnose \texorpdfstring{$D$}{D} infected persons}

	In what follows, we use the word \emph{position} as in the \emph{winning
	positions} in the combinatorial game theory, especially in the theory of the
	game of Nim \cite{BCG01}.  In this context, executing a tropical protocol is
	like playing a game against mother nature.  A position is a “current state”
	when we are halfway toward completely understanding everyone's Ct values.
	We \emph{win} the game by picking out all infected individuals in a limited
	amount of \emph{moves}.  A move is analogous to one single test if we care
	about the total number of tests, or to a batch of parallel tests if we care
	about the number of rounds.
	
	Denote by a tuple $(D^{(1)},D^{(2)},\dotsc,D^{(Π)})$ a position where
	\begin{itemize}
		\item	there are $Π$ piles of people;
		\item	the $π$th pile contain $N^{(π)}$ persons whose Ct values
				are denoted by $x^{(π)}₁,\dotsc,x^{(π)}_{N^{(π)}}$;
		\item	the first pile contains $D^{(1)}≥0$ infected persons;
		\item	for each $π≥2$, the $π$th pile contains
				$D^{(π)}≥1$ infected persons;  and
		\item	for each $π≥2$, we know the search results
				$a^{(π)}≔\min_jx^{(π)}_j+j$ and
				$b^{(π)}≔\min_jx^{(π)}_j+(N^{(π)}+1-j)$
				of the $π$th pile.
	\end{itemize}

	There are two moves that evolves a position into another position.

	If $Π=1$, we perform searching \cref{mat:7DTa} on the one and only pile.
	Denote the test result by $a^{(1)}$.  If $a^{(1)}$ is negative, the protocol
	terminates and reports everyone healthy.  If $a^{(1)}$ is positive, perform
	searching \cref{mat:7DTb} and denote the result by $b^{(1)}$.  Of this pile
	of people we now know the “$a$” and “$b$”; we migrate these people to the
	second pile and leave the first pile empty.  That is, we evolve the position
	$(D^{(1)})$ into the position $(0,D^{(1)})$.
	
	If $Π≥2$, we focus on the second pile. Of this pile we already know
	$a^{(2)}$ and $b^{(2)}$.  Compute $j^{(2)}≔a^{(2)}+b^{(2)}-N^{(2)}-1$.  We
	perform the verify--prefetching \cref{mat:3DTc} with $j≔j^{(2)}$ and denote
	the result by $c^{(2)}$.  By \cref{lem:monotony}, $c^{(2)}$ is equal to or
	greater than $(a^{(2)}+b^{(2)}-N^{(2)}-1)/2$.

	If $c^{(2)}$ is equal to $(a^{(2)}+b^{(2)}-N^{(2)}-1)/2$.  then $j^{(2)}$ is
	infected with Ct value $c^{(2)}$.  The remaining of the second pile, with
	$D^{(2)}-1$ patients remained to be found, is merged with the first pile.
	Now the position $(D^{(1)},\dotsc,D^{(Π)})$  is evolved into
	$(D^{(1)}+D^{(2)}-1,D^{(3)},\dotsc,D^{(Π)})$.

	If $c^{(2)}$ is greater than $(a^{(2)}+b^{(2)}-N^{(2)}-1)/2$, then we know
	\begin{align*}
		a^{(2)}	&	=\min_{i≤j^{(2)}}x^{(2)}_i+i,	\\
		b^{(2)}	&	=\min_{k>j^{(2)}}x^{(2)}_k+(N^{(2)}+1-k),	\\
		c^{(2)}	&	=\min_{i≤j^{(2)}}x^{(2)}_i+(j^{(2)}-i).
	\end{align*}
	For the first $j^{(2)}$ persons, we now know their “$a$” and “$b$”; for the
	last $N^{(2)}-j^{(2)}$ persons, we now know their “$b$”.  It suffices to
	query the latter's “$a$”:
	\[d^{(2)}≔\min_{k>j^{(2)}}x^{(2)}_k+(k-j^{(2)}).\]
	Consequently, we know the results of the searching matrices for the first
	half and the second half.  Suppose the first half contains $D^{(2,≤)}$
	infected persons and the second half contains $D^{(2,>)}$ infected persons.
	(We know nothing about $D^{(2,≤)}$ and $D^{(2,>)}$
	beyond that they are positive and sum to $D^{(2)}$.)
	Now the position is evolved from $(D^{(1)},\dotsc,D^{(Π)})$ into
	$(D^{(1)},D^{(2,≤)},D^{(2,>)},D^{(3)},\dotsc,D^{(Π)})$.

	Knowing how positions evolve, we estimate the cost.

	\begin{thm}\label{thm:nim}
		The protocol specified in this subsection is
		a $(3D+1)$-$(3D+1,N,D)$-tropical protocol.
	\end{thm}

	\begin{proof}
		Declare a budget function
		\[𝒯(D^{(1)},D^{(2)},\dotsc,D^{(Π)})≔3-2Π+3∑_{π=1}^ΠD^{(π)}.\]
		It is clear that when $Π=1$, the $𝒯$-formula of a singleton collapses
		to $𝒯(D^{(1)})=3-2+3D^{(1)}=3D^{(1)}+1$.  Hence the theorem will be
		proved if we can show that each position $(D^{(1)},\dotsc,D^{(Π)})$
		requires at most $𝒯(†that position†)$ tests.

		We employ induction on the $𝒯$-values, the budgets, of the
		positions.  We will show that whenever one position is evolved
		into another position, the number of tests spent is at most
		$𝒯(†former position†)-𝒯(†latter position†)$.  That way, the budget
		will be ever decreasing so it always satisfies the induction hypothesis
		afterwards.  But never would the budget go below zero before we finish
		picking out all infected individuals.

		Base case---$Π=1$ and $D^{(1)}=0$:  We spend one tests on $a^{(1)}$, the
		“$a$” of the first and only pile.  After getting a negative result we
		conclude that everyone is healthy.  Since
		\[𝒯(0)=1,\]
		the cost meets the budget for the base case.

		Now suppose that all positions whose $𝒯$-value fall below $T$ can be
		cleared before the budget runs out.  Suppose that
		$(D^{(1)},\dotsc,D^{(Π)})$ is a position with $𝒯$-value $T$. We
		want to show that it can be cleared before the budget runs out.

		Induction step, case one---$Π=1$ and $D^{(1)}≥1$:  We spend two tests on
		$a^{(1)}$ and $b^{(1)}$, the “$a$” and “$b$” of the first pile.  After
		that we relabel the first pile as the second pile, evolving $(D^{(1)})$
		into $(0,D^{(1)})$.  Since
		\[𝒯(D^{(1)})=2+𝒯(0,D^{(1)}),\]
		the cost meets the budget for induction step, case one.

		Induction step, case two---$Π≥2$ and
		$c^{(2)}=(a^{(2)}+b^{(2)}-N^{(2)}-1)/2$: We spend one test on $c^{(2)}$
		and confirmed that $j^{(2)}$ is infected with Ct value $c^{(2)}$.
		By doing so, we evolve $(D^{(1)},\dotsc,D^{(Π)})$ into
		$(D^{(1)}+D^{(2)}-1,D^{(3)},\dotsc,D^{(Π)})$.  Since
		\begin{align*}
			\qquad&\kern-2em
			𝒯(D^{(1)},\dotsc,D^{(Π)})	\\
			&	=1+𝒯(D^{(1)}+D^{(2)}-1,D^{(3)},\dotsc,D^{(Π)}),
		\end{align*}
		the cost meets the budget for induction step, case two.

		Induction step, case three---$Π≥2$ and
		$c^{(2)}>(a^{(2)}+b^{(2)}-N^{(2)}-1)/2$: We spend two tests on
		$c^{(2)}$ and $d^{(2)}$ to complete our knowledge of the “$a$” and “$b$”
		of the first $j^{(2)}$ persons and the last $N^{(2)}-j^{(2)}$ persons.
		We evolve $(D^{(1)},\dotsc,D^{(Π)})$ into
		$(D^{(1)},D^{(2,≤j)},D^{(2,>j)},D^{(3)},\dotsc,D^{(Π)})$ by doing so.
		Since
		\begin{align*}
			\qquad&\kern-2em
			𝒯(D^{(1)},\dotsc,D^{(Π)})	\\
			&	=2+𝒯(D^{(1)},D^{(2,≤j)},D^{(2,>j)},D^{(3)},\dotsc,D^{(Π)}),
		\end{align*}
		the cost meets the budget for induction step, case three.  This is the
		last piece of the induction and hence completes the proof.
	\end{proof}

	A generalization of \cref{thm:nim} to a delay-limited situation
	is the following.  We defer the proof until \cref{pf:deep}.
	
	\begin{thm}[Deep searching]\label{thm:deep}
		Let $T≔4D⌈㏒_ℓN⌉+1$.  There is a
		$T$-$(T,N,D)$-tropical protocol within maximum delay $ℓ$.
	\end{thm}

\section{Discussion}

\subsection{Open Problems}

	For the $D=2$ case with limited delay (\cref{thm:bite}), it is
	straightforward to handle one patient and two unequally infected patients.
	For two equally infected patients, we use $Q^{(k)}$ to obtain extra
	information.  Can we simplify $Q^{(k)}$?  The goal is to find a weaker
	notion of BITE that is still an inductive invariant.

	For general $D$ in nonadaptive case (\cref{thm:cycle}), we used exponential
	$ℓ$ but remarked that a polynomial $ℓ$ should be possible.  Is it?  Also,
	do there exist structural constructions that, more or less, generalize
	\cref{thm:bite} to general $D$?

	For the adaptive strategy with limited delay (\cref{thm:deep}), the main
	term of $T$ is $4D㏒₂N$.  From an information-theoretical perspective, the
	main term should be $D㏒₂N$.  Is this achievable?  Are there tradeoffs
	between $T$ the number of total tests and $R$ the number of rounds?

	Recall that PCR only runs for $40$ cycles in real life. Thus, if there is a
	Ct $35$ specimen delayed by $15$ cycles, the expected Ct value is $50$ but
	we only see $40$.  This is a false negative result. Can we increase the
	maximum delay to the extent where false negative starts showing up but we
	still benefit from it?  Note that in this case, the decoder must perform
	more a complicated pattern matching, one that treats $40$ as a wildcard that
	can possibly be $40$, $45$, $50$, or infinity.

\subsection{Concern of design}

	Hong et al.\ \cite{HDLCD21} suggested using factorizations of a complete
	hypergraph to generate balanced pooling designs.  Here, \emph{balanced}
	means that every person appears in the same number of tubes and every tubes
	receives (almost) the same number of persons.\footnote{ In the block design
	context, every block having the same number of vertices is called
	\emph{uniform} or \emph{proper}; every vertex appearing in the same number
	of blocks is called \emph{balanced}.  In the hypergraph context, the former
	is called \emph{uniform}; the latter is called \emph{regular}.  They want
	both properties.} They argued that those conditions make pooling more
	consistent.  We, while agreeing with their argument, want to add that there
	are other ways to achieve the same goal.  Kirkman systems, Steiner systems,
	BIBDs, and constant weight codes are candidates that sound equally good.  To
	be more specific, we believe that one should (also) optimize for the
	probability that a block $B$ is covered by the union of a small number of
	other blocks $Z₁,\dotsc,Z_D$ (cf.\ superimposed codes, $D$-cover-free
	families).

\subsection{Origin of masking}

	It is worth mentioning that Hwang and Xu \cite{HX87} once published a
	variant of group testing where there are two infected people, one heavily
	infected and the other lightly infected.  The testing result is quantified
	by three possibilities:  (i) The tested pool contains the heavily infected
	person.  (ii) The tested pool contains the lightly infected person, but not
	the heavily one.  (iii) The tested pool contains neither infected person.
	Notice that (i), (ii), and (iii) can be thought as Ct values $1$, $2$, and
	$3$, respectively, while the infected persons have Ct values $1$ and $2$.
	Hwang and Xu's problem formulation already suggested that the presence of
	the heavily infected complicates the identification of the lightly infected.
	The group testing scheme proposed in \cref{sec:pipet,sec:bite} and
	\cref{pf:bite} can be interpreted as a solution to a generalization of their
	setup.

\section{Acknowledgment}

	The authors thank Gerry Myerson\footnote{
		\url{https://math.stackexchange.com/a/4203073}},
	Venkatesan Guruswami, and
	Ching-Hung Hsieh\footnote{
		Kuang Tien General Hospital, Taiwan.}
	for pointing out references.
	The authors thank Chu-Lan Kao\footnote{
		\url{https://orcid.org/0000-0001-8091-9946}}
	for recommending importance sampling and sharing statistics insights.
	The authors thank Chih-Yang Hsia\footnote{
		\url{https://scholar.google.com/citations?user=Qrnp5CMAAAAJ}}
	for programming aids.
	The authors thank Facebook user kiwi.qin
	for discussions about the $D=2$ case.

\appendices
\crefalias{section}{appendix}

\section{Projective Space (Proof of Theorem~\ref{thm:shell})}\label{pf:shell}

	In this appendix, we factor in the restriction that one often
	cannot wait for arbitrarily long delays but can afford more tests.
	This appendix shows how to trade $T$ for $ℓ$ and proves \cref{thm:shell}.
	
\subsection{Three tests for higher volume}

\begin{figure}
	\def\tikz@install@auto@anchor@tip{\def\tikz@do@auto@anchor{
		\pgfmathsetmacro\tikz@anchor{atan2(-\pgf@y,-\pgf@x)}}}
	\tdplotsetmaincoords{90-23.97565}{120.973882} 
	\def\camerax{3}\def\cameray{3}\def\cameraz{3}\def\cameras{6.666}\def\camerad{10}
	\def\savepoint#1#2#3#4{
		\pgfmathsetmacro\xx{#1}
		\pgfmathsetmacro\yy{#2}
		\pgfmathsetmacro\zz{#3}
		\pgfpointxyz{\xx}{\yy}{\zz}
		\xdef#4{\noexpand\pgfpoint{\the\pgf@x}{\the\pgf@y}}
	}
	\def\savepoints#1#2#3#4#5{
		\savepoint{#1}{#2}{#3}{#4}
		\pgfpointxyz{\xx+\t}{\yy+\t}{\zz+\t}
		\xdef#5{\noexpand\pgfpoint{\the\pgf@x}{\the\pgf@y}}
	}
	\def\r{.5}\def\s{.1}
	\def\pencil(#1,#2,#3)+#4 #5;{
		\pgfmathsetmacro\x{#1}
		\pgfmathsetmacro\y{#2}
		\pgfmathsetmacro\z{#3}
		\def\t{#4}									
				\savepoints{\x}{\y-\r}{\z+\r}\Pb\Qb\savepoints{\x-\r}{\y}{\z+\r}\Pa\Qa
		\savepoints{\x+\r}{\y-\r}{\z}\Pc\Qc				\savepoints{\x-\r}{\y+\r}{\z}\Pf\Qf
				\savepoints{\x+\r}{\y}{\z-\r}\Pd\Qd\savepoints{\x}{\y+\r}{\z-\r}\Pe\Qe
		\begin{pgfonlayer}{background}
			\pgfpathmoveto\Pa\pgfpathlineto\Pb\pgfpathlineto\Pc
			\pgfpathlineto\Pd\pgfpathlineto\Pe\pgfpathlineto\Pf\pgfclosepath
			\pgfusepath{fill,stroke}
		\end{pgfonlayer}		
				\pgfpathmoveto\Pb\pgfpathlineto\Qb\pgfpathmoveto\Pa\pgfpathlineto\Qa
		\pgfpathmoveto\Pc\pgfpathlineto\Qc				\pgfpathmoveto\Pf\pgfpathlineto\Qf
				\pgfpathmoveto\Pd\pgfpathlineto\Qd\pgfpathmoveto\Pe\pgfpathlineto\Qe
		\pgfusepath{stroke}
		\pgfpathmoveto\Qa\pgfpathlineto\Qb\pgfpathlineto\Qc
		\pgfpathlineto\Qd\pgfpathlineto\Qe\pgfpathlineto\Qf\pgfclosepath
		\pgfusepath{fill,stroke}
		\pgfpointdiff{\pgfpointxyz{\x+\t+.35355}{\y+\t-.35355}{\z+\t}}
					{\pgfpointxyz{\x+\t-.35355}{\y+\t+.35355}{\z+\t}}
		\xdef\ux{\the\pgf@x}\xdef\uy{\the\pgf@y}
		\pgfpointdiff{\pgfpointxyz{\x+\t+.20412}{\y+\t+.20412}{\z+\t-.40825}}
					{\pgfpointxyz{\x+\t-.20412}{\y+\t-.20412}{\z+\t+.40825}}
		\xdef\vx{\the\pgf@x}\xdef\vy{\the\pgf@y}
		\pgflowlevelobj{
			\pgftransformcm{\ux/28.4527}{\uy/28.4527}{\vx/28.4527}{\vy/28.4527}
							{\pgfpointxyz{\x+\t}{\y+\t}{\z+\t}}
		}{
			\draw(0,0)node[black]{#5};
		}
	}
	$$\tikz[tdplot_main_coords]{
		\begin{pgfonlayer}{background}
			\draw[line width=.6,nodes={pos=1,auto=tip}]
				(0,0,0)edge[->]node{$a$}(6,0,0)
				(0,0,0)edge[->]node{$b$}(0,6,0)
				(0,0,0)edge[->]node{$c$}(0,0,6)
			;
			\pgfsetfillcolor{PMS116}\pgfsetstrokecolor{PMS3015}
		\end{pgfonlayer}
		\pgfsetfillcolor{PMS116}\pgfsetstrokecolor{PMS3015}
							\pencil(0,0,3)+3 {V};
			\pencil(3,0,3)+3 {W};			\pencil(0,3,3)+3 {U};
							\pencil(0,0,0)+6 {T};
			\pencil(3,0,0)+3 {X};			\pencil(0,3,0)+3 {Z};
							\pencil(3,3,0)+3 {Y};
		\begin{pgfonlayer}{background}
			\draw[line width=.6,black]
				(0,0,0)edge(1,0,0)edge(0,1,0)edge(0,0,1)
				(0,0,3)--(0,0,4)(0,3,0)--(0,4,0)(3,0,0)--(4,0,0)
			;
		\end{pgfonlayer}
	}$$
	\caption{
		Thr generalization of \cref{fig:triton} to three tests.
	}\label{fig:pencil}
\end{figure}

	The following is a three-test seven-person example.  Name the persons T, U,
	V, W, X, Y, Z and let $t,u,v,w,x,y,z$ be their Ct values.  Encoding:
	\[\bma{
		a	\\	b	\\	c
	}≔\bma{
		0	&	0	&	0	&	7	&	7	&	7	&	0	\\
		0	&	7	&	0	&	0	&	0	&	7	&	7	\\
		0	&	7	&	7	&	7	&	0	&	0	&	0	
	}⊙\bma{
		t	\\	u	\\	⋮	\\	z
	}\]
	Decoding:  \Cref{fig:pencil} generalizes \cref{fig:triton} and illustrate
	how to map the test results back to who is infected and how infected they
	are.

	For a finite $ℓ$, the columns of a schedule matrix are vectors in
	$\{0,1,\dotsc,ℓ,∞\}^3$ with at least one $0$.  The number of such lattice
	points is $(ℓ+2)³-(ℓ+1)³=3ℓ²+9ℓ+7$.  For $ℓ=0$, $1$, and $2$, the number of
	lattice points are $7$, $19$, and $37$, which are $2$x, $6$x, and $12$x
	increases in throughput, respectively.

\subsection{Asymptote of one infection}

	Fix an upper bound on delay $ℓ≥0$.  How fast can
	$N$ grow if the number of tests $T$ approaches infinity?
	Clearly we want to select, for each person, a delay column
	\[𝛅=\bma{
		δ₁		\\	⋮	\\	δ_T
	}∈\{0,1,\dotsc,ℓ,∞\}^{T×1}\]
	such that the straight lines
	\[\left\{𝛅⊙\bma{x}=\left[\sma{
		δ₁+x	\\	⋮	\\	δ_T+x
	}\right]\;\middle|\;x∈ℝ\right\}⊆(ℝ∪\{∞\})^{T×1}\]
	are disjoint (far away) from each other.
	Note that every line contains one and only one delay column
	\[˜{𝛅}≔𝛅⊙\bma{
		-\min(𝛅)
	}=\bma{
		δ₁-\min(𝛅)	\\	⋮	\\	δ_T-\min(𝛅)
	}\]
	that has at least one zero entry and no negative entries.
	This means that every line passes one and only one point in the “shell”
	\[\Sha≔\{0,1,\dotsc,ℓ,∞\}^{T×1}＼\{1,\dotsc,ℓ,∞\}^{T×1}.\]
	$\Sha$ has cardinality $(ℓ+2)^T-(ℓ+1)^T≈Tℓ^{T-1}$.
	We are ready to prove \cref{thm:shell}.

	\begin{proof}[Proof of \cref{thm:shell}]
		To see $N≤(ℓ+2)^T-(ℓ+1)^{T×1}$, observe that every column vector in
		$\{0,1,\dotsc,ℓ,∞\}^T$ is congruent to a column vector in $\Sha$ modulo
		\[\bma{
			1	\\	⋮	\\	1
		}.\]
		To obtain a tropical code that meets the bound $N=(ℓ+2)^T-(ℓ+1)^T$,
		use $\Sha$ per se or use a collection of column vectors
		that congruent to different column vectors in $\Sha$.
	\end{proof}

\section{Larger BITE (Proof of Theorem \ref{thm:bite})}\label{pf:bite}

	Let us first define the schedule matrices.  The proof follows.

	Fix a $(t,n,2)$-tropical code $S$.  Define $S^{(k)}$ to be this $kt×n^k$
	matrix
	\[S^{(k)}=\bma{
		𝟏_{1×n^{k-1}}⊗S	\\
		𝟏_{1×n^{k-2}}⊗S⊗𝟏_{1×n}	\\
		⋮	\\
		𝟏_{1×n}⊗S⊗𝟏_{1×n^{k-2}}	\\
		S⊗𝟏_{1×n^{k-1}}
	}\]
	for all $k≥2$.  Also define $Q^{(k)}$ to be this $(2k-3)×n^k$ matrix
	\[\arraycolsep3pt
	ω\left(\bma{
		1	&	⋯	&	β₁^{k-1}	\\
		⋮	&	⋰3	&	⋮	\\
		1	&	⋯	&	β_{2k-3}^{k-1}
	}\bma{	
		1	&	⋯	&	α₁^{k-1}	\\
		⋮	&	⋰3	&	⋮	\\
		1	&	⋯	&	α_k^{k-1}
	}^{-1}\bma{
		𝟏_{1×n^{k-1}}⊗M	\\
		⋮	\\
		M⊗𝟏_{1×n^{k-1}}
	}\right)\]
	where $ω：𝔽_q→\{0,\dotsc,q-1\}$ is a look-up bijection that applies to
	matrices entry-wisely,
	\[M=ω^{-1}\left(\bma{
		0	&	⋯	&	n-1
	}\right),\]
	$α₁,\dotsc,α_k$ and $β₁,\dotsc,β_{2k-3}$ are distinct elements in $𝔽_q$,
	and $q$ is the smallest prime power $≥\max(3k-3,n)$.
	We now use
	\[\bma{
		S^{(k)}	\\	Q^{(k)}
	}\label{mat:SkQk}\]
	to prove \cref{thm:bite}.

	\begin{proof}[Proof of \cref{thm:bite}]
		First of all, if everyone is healthy or there is merely one
		infected person, the situation will be trivial.  Hereafter we
		assume that there are two infected people.  Let their indices be
		$Y=y₁+y₂n+\dotsb+y_kn^{k-1}$ and $Z=z₁+z₂n+\dotsb+z_kn^{k-1}$ in
		their $n$-ary expansions.  Let their Ct values be $y_*$ and $z_*$,
		respectively. Then the first $t$ tests teach us $\{(y₁,y_*),(z₁,z_*)\}$
		or $\{(z₁,\min(y_*,z_*))\}$ if $y₁≠z₁$ or $y₁=z₁$, respectively.
		In general, the $(jt-t+1)$th to the $(jt)$th tests teach us
		$\{(y_j,y_*),(z_j,z_*)\}$ or $\{(z_j,\min(y_*,z_*))\}$ depending
		on whether the digits differ or not.

		The next step is to sort $y₁,z₁,\dotsc,y_k,z_k$ into two piles,
		$y₁,\dotsc,y_k$ and $z₁,\dotsc,z_k$, so that we can recover $Y$ and
		$Z$.  This can be done when two patients assume different Ct values,
		$y_*≠z_*$, in which case we know the digits of $Y$ are those that
		associate to $y_*$ and the digits of $Z$ are those that associate to
		$z_*$.  On the other hand, if we only see one Ct value the whole time,
		then $y_*=z_*$.  We will utilize the last $2k-3$ tests as the checksums
		of a systematic Reed--Solomon code to help sorting.  Here is how.

		Without loss of generality, we may assume $y_*=0=z_*$.
		Then the decoding boils down to the following task:  Suppose
		\[Y≔ω^{-1}(y₁,\dotsc,y_k,u₁,\dotsc,u_{2k-3})\label{cod:y}\]
		and
		\[Z≔ω^{-1}(z₁,\dotsc,z_k,v₁,\dotsc,v_{2k-3})\label{cod:z}\]
		are two codewords of a $[3k-3,k,2k-2]$-Reed--Solomon code.  Suppose
		that we know $\{y₁,z₁\}$ to $\{y_k,z_k\}$.  Suppose we also know
		$\min(u₁,v₁)$ to $\min(u_{2k-3},v_{2k-3})$.  To recover $Y$ and $Z$,
		make a guess of two codewords
		\[A≔ω^{-1}(a₁,\dotsc,a_k,e₁,\dotsc,e_{2k-3})\]
		and
		\[B≔ω^{-1}(b₁,\dotsc,b_k,f₁,\dotsc,f_{2k-3})\]
		such that $\{a_i,b_i\}=\{y_i,z_i\}$ for $1≤i≤k$ and
		$\min(e_j,f_j)=\min(u_j,v_j)$ for $1≤j≤2k-3$.
		If remains to show that $\{A,B\}=\{Y,Z\}$.

		Let $I_{AY}$ be the set of coordinates $i$ where $a_i=y_i$.
		Let $J_{AY}$ be the set of coordinates $j$ where
		$e_j=\min(e_j,f_j)=\min(u_j,v_j)=u_j$.  Define $I_{AZ}$,
		$J_{AZ}$, $I_{BY}$, $J_{BY}$, $I_{BZ}$, and $J_{BZ}$ similarly.
		Then $|I_{AY}|+|I_{AZ}|+|I_{BY}|+|I_{BZ}|≥2k$ and
		$|J_{AY}|+|J_{AZ}|+|J_{BY}|+|J_{BZ}|≥2k-3$.  Hence at least one of
		$|I_{AY}|+|J_{AY}|$, $|I_{AZ}|+|J_{AZ}|$, $|I_{BY}|+|J_{BY}|$, and
		$|I_{BZ}|+|J_{BZ}|$ is $≥k$.  Say $|I_{BZ}|+|J_{BZ}|≥k$.  Then the
		Reed--Solomon code being $[3k-3,k,2k-2]$ forces $B=Z$.  This then
		forces $A=Y$ and we finish proving that schedule \cref{mat:SkQk}
		is a valid tropical code.

		Now assume that $S$ BITEs.  Then $n/4>2^t$.  This implies
		$n^k/4>2^{kt+2k-3}$, hence schedule \cref{mat:SkQk} BITEs.
	\end{proof}

\section{Deep Searching (Proof of Theorem \ref{thm:deep})}\label{pf:deep}

	Let $ℓ$ be the largest available delay.  We want to prove that the number
	of tests needed is at most $T=4D⌈㏒_ℓN⌉+1$.  Here is an overview of the
	strategy.
	
	Instead of \cref{mat:hinge}, we can only afford this schedule matrix
	\[\bma{
		1⋯1	&	•••	&	ℓ⋯ℓ	\\
		ℓ⋯ℓ	&	•••	&	1⋯1		
	}\label{mat:logNDab}\]
	where each unique column repeats $≈N/ℓ$ times.  In other words, we divide
	$N$ people into $ℓ$ piles, treat each pile as one person, and apply
	searching \cref{mat:7DTa,mat:7DTb}.
	
	Now suppose that the diff-lay $a-b$ points to the $j$th pile.  That is,
	$j≔(a-b+ℓ-1)/2$.  Similar to \cref{lem:dichotomy}, there are three
	possibilities.
	\begin{enumerate}
		\item[(i)] 	The $j$th pile contains a patient and her Ct value is $a$.
		\item[(ii)]	The $j$th pile is healthy.  Instead, the first $j-1$ piles
					contains $≥1$ patient and the last $ℓ-j$ piles contains
					$≥1$ patient.
		\item[(iii)]The $j$th pile does contain some patients but their Ct
					values are $>a$.  Meanwhile, the first $j-1$ piles
					contain $≥1$ patient and the last $ℓ-j$ piles contain
					$≥1$ patient.
	\end{enumerate}
	Apply the verify--prefetching \cref{mat:3DTc}.
	\[\bma{
		c^♯
	}≔\bma{
		j-1⋯j-1	&	•••	&	1⋯1	&	0⋯0	\\
	}⊙\bma{
		x_1	\\	⋮	\\	x_{÷{jN}{ℓ+1}}
	}\]
	If $c^♯$ equals $(a+b-ℓ-1)/2$, then (i) is the case and we can reduce
	our scope to the $j$th pile, which is one order of magnitude smaller.
	If $c^♯$ is greater then $(a+b-ℓ-1)/2$, then (ii) is the case and
	the remaining population splits into two super-piles of piles
	\[\bma{
		x₁	\\	⋮	\\	x_{÷{(j-1)N}{ℓ}}
	}† and †\bma{
		x_{÷{jN}{ℓ}+1}	\\	⋮	\\	x_N
	}\]
	each having at least one infected person.  Now it suffices
	to apply the recursive algorithm to them separately.

	Overall, our strategy is a digit-by-digit $ℓ$-ary searching.  In sunny
	days we can confirm a digit and reduce the candidates to a smaller
	population.  In rainy days we split the population into two halves.

	\begin{proof}[Proof of \cref{thm:deep}]
		Let there be $N$ persons to be tested.  Partition them into $ℓ$
		almost-equal piles, treat them as $ℓ$ persons and apply \cref{thm:nim}.
		If there is no patient, the protocol will end the moment the first test
		comes out negative.  Hereafter, we let $D≥1$.  Let there be $D^•$
		infected piles.  Then it takes $3D^•+1≤4D$ test to single out the
		infected piles.  To each and every infected pile, apply \cref{thm:nim}
		recursively.  That will single out some infected sub-piles, followed by
		some infected sub-sub-piles.  And so on and so forth.  It takes at most
		$4D⌈㏒_ℓN⌉$ tests to split the population down to atomic individuals.
		This finishes the proof.
	\end{proof}

\section{Probabilistic Model and Simulation}\label{app:simulate}

	Alongside the development of the theoretical aspects of tropical group
	testing, we devote the very last appendix to benchmarking the performance
	of matching and delaying in a practical setup.  These simulations justify
	why tropical arithmetic---minimum and addition---are good approximation
	of the reality.  They also clarify what design elements makes a scheme
	clinically-usable.

\subsection{The prior}

	Let $p∈[0,1]$ be the prevalence rate.  Let $N$ denote the number of people
	to be tested.  For every individual $j∈[N]$, we toss a Bernoulli coin with
	mean $p$.  If the outcome of the toss is $0$, then individual $j$ is healthy
	and her Ct value $x_j$ is set to be $99$; if the outcome of the toss is $1$,
	then individual $j$ is infected and her Ct value $x_j$ is drawn continuously
	uniformly from the interval $[16,32]$ \cite{JBMVSBBTSSKMSZHKSECD21,
	JKKLGLM21, JGJO20, BMR21, BHF21}.  Note that $x_j$ is not necessarily an
	integer.  Under this setup, the viral load $x_j$ is set to be $2^{-x_j}$.
	Denote by
	\[2^{-𝐱}≔\bma{
		2^{-x₁}	\\	⋮	\\	2^{-x_N}
	}\]
	the column vector that represents the number of virus particles of the
	specimens from each of the $N$ individuals.
	
\subsection{The encoder}
	
	Given a block design $ℱ$ and a bijection $j：ℱ→[N]$, we will construct the
	schedule matrix $S$ as we have always been:
	\[S_{tj(B)}≔\cas{
		δ	&	if $t∈B$,	\\
		∞	&	if $t∉B$,	
	}\]
	where $δ=ℓ·†Bernoulli†(1/2)∈\{0,ℓ\}$ is a random delay generated unbiasedly
	and independently.   Let $2^{-S}$ be the result of, entry-wisely, $2$ raised
	to the power of $-S$.  That is,
	\[2^{-S}≔\bma{
		2^{-S_{11}}	&	⋯	&	2^{-S_{1N}}	\\
		⋮			&	⋰2	&	⋮			\\
		2^{-S_{T1}}	&	⋯	&	2^{-S_{TN}}
	}.\]
	One sees that we implement delaying by $δ$ cycles by diluting by a factor of
	$2^δ$.

	Now perform the PCR tests
	\[𝐯≔2^{-S}2^{-𝐱}.\]
	Its $t$th row, denoted by $v_t$, is the viral load of the $t$th tube given
	the delay schedule $S$ and the viral loads $2^{-𝐱}$.  Let, entry-wisely,
	\[𝐜≔\min(40,⌊-㏒₂(𝐯)⌋)=\bma{
		\min(40,⌊-㏒₂(v₁)⌋)	\\
		⋮					\\
		\min(40,⌊-㏒₂(v_T)⌋)
	}\]
	be the collection of the capped, integral Ct values of the tubes.  Its
	$t$th row is denoted by $c_t$.  The decoder will be given $S$ and $𝐜$, from
	which it shall determine who are infected and how severe their infection
	are.

\subsection{The decoder}

\tikzset{
	claw/.style={xscale=-1.4142,yscale=.7071,rotate=-90},
	test/.style={draw,inner sep=3},
	person/.style={ellipse,draw,inner sep=1,align=center},
	so/.style=PMS3015,
	delta/.style=auto,
}
\begin{figure}
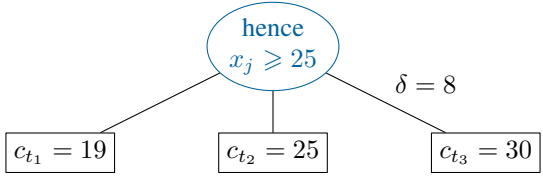

	$$\tikz{
		\draw[claw]
			                                           (1,2)node(t3)[test]{$c_{t₁}=19$}
			(-1,0)node(p1)[person,so]{hence\\$x_j≥25$} (1,0)node(t2)[test]{$c_{t₂}=25$}
			                                           (1,-2)node(t1)[test]{$c_{t₃}=30$}
			(p1)--(t3)
			(p1)--(t2)
			(p1)--node[delta]{$δ=8$}(t1)
		;
	}$$
	\caption{
		Phase I---underestimate:  Suppose that the $j$th person participates in
		tests $t₁,t₂,t₃∈[T]$.  Then $x_j$ has lower bounds $c_{t₁}-S_{t₁j}$,
		$c_{t₂}-S_{t₂j}$, and $c_{t₃}-S_{t₃j}$.  They are $19$, $25$, and
		$30-8=22$, respectively. The strongest bound is $x_j≥25$ so we let
		$u_j≔25$.
	}\label{fig:PhaseI}
\end{figure}

	Phase I---underestimate:  For any $t∈[T]$ and $j∈[N]$, the amount of virus
	in the $t$th tube is at least what was contributed by the $j$th person,
	i.e., $v_t≥2^{-S_{tj}}2^{-x_j}$.  This yields inequalities
	$c_t≤-㏒₂(v_t)≤S_{tj}+x_j$.  We therefore define
	\[u_j≔\max_tc_t-S_{tj}\]
	to be an underestimate of $x_j$.  See \cref{fig:PhaseI} for an example.

\begin{figure}
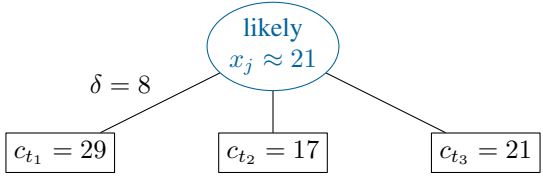

	$$\tikz{
		\draw[claw]
			                                           (1,2)node(t3)[test]{$c_{t₁}=29$}
			(-1,0)node(p1)[person,so]{likely\\$x_j≈21$} (1,0)node(t2)[test]{$c_{t₂}=17$}
			                                           (1,-2)node(t1)[test]{$c_{t₃}=21$}
			(p1)--node[delta,']{$δ=8$}(t3)
			(p1)--(t2)
			(p1)--(t1)
		;
	}$$
	\caption{
		Phase II---match:  If the $j$th person is the main contributor of some
		two tubes $t₁,t₂∈[T]$, we will see $c_{t₁}-S_{t₁j}=x_j=c_{t₂}-S_{t₂j}$.
		Conversely, whenever we see that the greatest two $c_t-S_{tj}$ coincide,
		we can guess with confidence that $x_j=u_j$.  In this figure,
		$c_{t₁}-S_{t₁j}=29-8=21=c_{t₃}-S_{t₃j}$, so $x_j$ is likely $21$.
	}\label{fig:PhaseII}
\end{figure}

	Phase II---match:  For any person with Ct value $x_j$ and any tube $t$,
	either $j$ has contributed the majority of the virus and hence
	$c_t≈S_{tj}+x_j$, or someone else contributed significantly more and
	$c_t≪S_{tj}+x_j$.  Thinking backward, each $c_t-S_{tj}$ is either
	$x_j$ or less than that.  We thus look for person $j$ where
	\[\bigl|\{t｜c_t-S_{tj}=u_j\}\bigr|≥2.\]
	For every such $j$, the decoder declares that $j$ is infected and the
	inferred Ct value is $u_j$.  See \cref{fig:PhaseII} for an example.

\begin{figure}
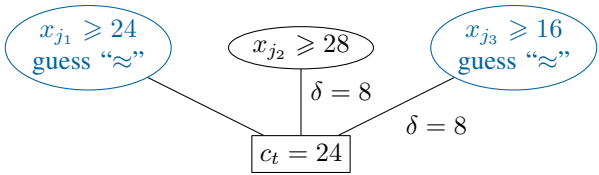

	$$\tikz{
		\draw[claw]
			(-1,2)node(p3)[person,so]{$x_{j₁}≥24$\\guess “$≈$”}
			(-1,0)node(p2)[person]{$x_{j₂}≥28$}            (1,0)node(t1)[test]{$c_t=24$}
			(-1,-2)node(p1)[person,so]{$x_{j₃}≥16$\\guess “$≈$”}
			(p3)--(t1)
			(p2)--node[delta,pos=1/3]{$δ=8$}(t1)
			(p1)--node[delta]{$δ=8$}(t1)
		;
	}$$
	\caption{
		Phase III---explain:  If the $t$th tube is positive but not explained
		by any patients reported by the first two phases we will find in $t$'s
		participants the most likely people and report them infected.  In the
		figure, both $j₁$ and $j₃$ can contribute $2^{-24}$, which is what $t$
		has right now, so both are declared infected.
	}\label{fig:PhaseIII}
\end{figure}

	Phase III---explain:  Motivated by the SCOMP algorithm \cite{ABJ14}, we want
	to make sure that all positive tubes are explained by some positive person.
	For any unexplained tube, we look at its participants and look for the
	one(s) that could have been the main contributor.  In detail, recall that
	the $t$th tube has about $2^{-c_t}$ and the $j$th person could have
	contributed at most $2^{-S_{tj}-u_j}$.  Therefore, the more negative the
	deficit $c_t-S_{tj}-u_j$, the less likely $j$ is responsible for tube $t$.
	More concisely, we look for the subset of suspects
	\[\{j｜c_t-S_{tj}-u_j=\max_kc_t-S_{tk}-u_k\}.\]
	We will report everyone in this subset infected, each with her own $u_j$ as
	the speculated Ct value.  See \cref{fig:PhaseIII} for an example.

\begin{figure}
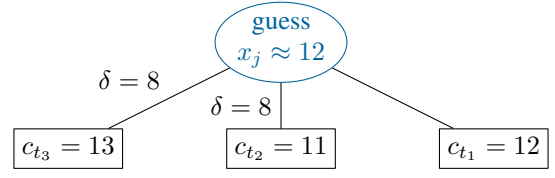

	$$\tikz{
		\draw[xscale=1.4142,yscale=.7071,rotate=-90]
			                                           (1,2)node(t3)[test]{$c_{t₁}=12$}
			(-1,0)node(p1)[person,so]{guess\\$x_j≈12$} (1,0)node(t2)[test]{$c_{t₂}=11$}
			                                           (1,-2)node(t1)[test]{$c_{t₃}=13$}
			(p1)--(t3)
			(p1)--node[delta,']{$δ=8$}(t2)
			(p1)--node[delta,']{$δ=8$}(t1)
		;
	}$$
	\caption{
		Phase IV---dark room:  Fix a $ρ$; say $ρ=14$.  If we see $u_j<ρ$,
		then the $j$th person could have had $x_j∈[ρ,32]$ and we would not
		notice due to the masking effect.  We declare that she is infected with
		Ct value $u_j$.  In general, the diagnose of each person is a function
		in $ρ$.  Set a low $ρ$ then she is infected; set a high $ρ$ then she is
		healthy (unless the previous phases found clear evidence of infection).
	}\label{fig:PhaseIV}
\end{figure}

	Phase IV---dark room:  Consider the following scenario.  A person is 
	infected with Ct value $31$.  However, against our favor, the tubes she is
	in have Ct values $11$, $12$, and $13$.  By no means we can infer whether
	she is infected or not.  For this type of “patients in the dark room”, we
	set a bar $ρ$ and report anyone whose underestimate $u_j$ is lower than $ρ$.
	This way, the set of people being diagnosed infected is a function in $ρ$.
	For lower $ρ$, fewer people are diagnosed infected, so the specificity is
	high, but the sensitivity is low.  For higher $ρ$, it is the other way
	around.  This $ρ$ parametrizes an \emph{receiver operating characteristic}
	(ROC) curve.
	
\subsection{The simulation result}

\newcount\accumulatenumplots
\pgfplotsset{
	cycle multiindex* list={
		PMS116,PMS3015,PMS1245,PMS3115,PMS144\nextlist
		mark=o,mark size=1.5\\mark=+\\mark=triangle\\mark=*,mark size=1\\mark=x\\
		every mark/.append style={yscale=-1},mark size=2,mark=triangle*\\\nextlist
	},
	legend image code/.code={
		\draw[mark repeat=2,mark phase=2,#1]plot coordinates{
			(-4pt,-6pt)(0pt,2pt)(8pt,6pt)
		};
	},
	every axis/.style={cycle list shift=\the\accumulatenumplots},
	percentage/.style={
		yticklabel=\pgfmathprintnumber\tick\%,
		xticklabel=\pgfmathprintnumber\tick\%,
		xticklabel style={rotate=90},
		max space between ticks=24
	},
	RoC/.style={
		percentage,
		xlabel=false positive rate ($1-{}$specificity),
		ylabel=true positive rate (sensitivity),
		legend pos=south east
	}
}
\pgfplotstableread{
	5fp    5sen   6fp    6sen   7fp    7sen   8fp    8sen   9fp    9sen   10fp   10sen
	0.463  99.323 0.805  98.907 1.273  98.394 1.873  97.793 2.627  97.088 3.531  96.314
	0.463  99.323 0.805  98.907 1.273  98.394 1.873  97.793 2.627  97.088 3.531  96.314
	0.464  99.325 0.808  98.91  1.278  98.4   1.879  97.801 2.636  97.099 3.545  96.329
	0.471  99.331 0.82   98.924 1.296  98.42  1.905  97.83  2.673  97.139 3.591  96.38 
	0.485  99.347 0.842  98.95  1.335  98.455 1.959  97.885 2.745  97.213 3.685  96.478
	0.51   99.367 0.883  98.989 1.397  98.513 2.048  97.97  2.865  97.33  3.843  96.637
	0.545  99.399 0.943  99.043 1.49   98.592 2.177  98.087 3.043  97.488 4.074  96.839
	0.596  99.443 1.027  99.109 1.611  98.695 2.351  98.228 3.277  97.686 4.376  97.102
	0.66   99.491 1.137  99.193 1.768  98.814 2.571  98.401 3.57   97.908 4.759  97.405
	0.741  99.545 1.269  99.286 1.967  98.955 2.839  98.594 3.924  98.167 5.214  97.737
	0.843  99.604 1.431  99.385 2.192  99.102 3.155  98.798 4.347  98.438 5.746  98.084
	0.961  99.668 1.619  99.488 2.46   99.263 3.518  99.003 4.824  98.719 6.344  98.444
	1.105  99.746 1.828  99.596 2.762  99.421 3.941  99.221 5.362  99.006 7.011  98.793
	1.256  99.821 2.069  99.703 3.11   99.575 4.406  99.435 5.955  99.286 7.747  99.131
	1.432  99.879 2.332  99.81  3.495  99.725 4.911  99.644 6.602  99.548 8.522  99.447
	1.634  99.944 2.635  99.899 3.922  99.865 5.474  99.821 7.307  99.784 9.365  99.728
	1.855 100.    2.983 100.    4.391 100.    6.09  100.    8.08  100.   10.291 100.   
}\tablepreval
\begin{figure}
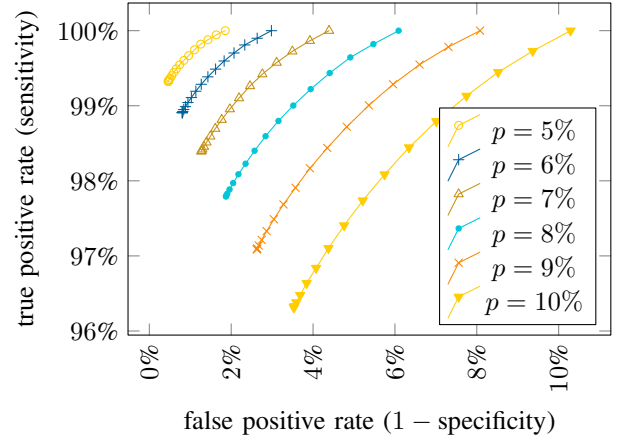

	$$\tikz{\linlin[RoC]{
		\addplot table[x=5fp,y=5sen]{\tablepreval};\addlegendentry{$p=5\%$}
		\addplot table[x=6fp,y=6sen]{\tablepreval};\addlegendentry{$p=6\%$}
		\addplot table[x=7fp,y=7sen]{\tablepreval};\addlegendentry{$p=7\%$}
		\addplot table[x=8fp,y=8sen]{\tablepreval};\addlegendentry{$p=8\%$}
		\addplot table[x=9fp,y=9sen]{\tablepreval};\addlegendentry{$p=9\%$}
		\addplot table[x=10fp,y=10sen]{\tablepreval};\addlegendentry{$p=10\%$}
	}}$$
	\global\advance\accumulatenumplots6
	\caption{
		Assume uniform Ct values on the interval $[16,32]$, $15×35$ Kirkman
		triple system, and no delay ($ℓ=0$).  We vary the prevalence rate $p$
		and plot the ROC curves.
	}\label{fig:ROCpreval}
\end{figure}

	\Cref{fig:ROCpreval} shows how different prevalence rates translate into
	performances.

\pgfplotstableread{
	8fp    8sen    12fp   12sen   16fp   16sen   20fp   20sen   24fp   24sen
	3.291  96.837  3.391  96.608  3.538  96.338  3.745  95.816  4.111  94.889
	3.291  96.837  3.391  96.608  3.538  96.338  3.745  95.816  4.111  94.892
	3.295  96.841  3.398  96.616  3.55   96.354  3.774  95.856  4.204  95.04 
	3.311  96.855  3.422  96.643  3.597  96.404  3.876  95.979  4.513  95.462
	3.344  96.89   3.475  96.695  3.694  96.502  4.086  96.208  5.093  96.161
	3.396  96.941  3.564  96.776  3.854  96.66   4.418  96.54   5.972  97.082
	3.478  97.017  3.691  96.894  4.084  96.866  4.884  96.96   7.128  98.093
	3.584  97.11   3.865  97.048  4.387  97.122  5.482  97.454  8.557  99.071
	3.727  97.231  4.086  97.225  4.767  97.413  6.21   97.991 10.275 100.   
	3.899  97.364  4.357  97.437  5.223  97.737  7.062  98.547    nan     nan
	4.105  97.516  4.684  97.667  5.755  98.089  8.022  99.078    nan     nan
	4.344  97.685  5.057  97.914  6.35   98.442  9.089  99.547    nan     nan
	4.616  97.869  5.477  98.172  7.016  98.804 10.282 100.       nan     nan
	4.926  98.066  5.936  98.437  7.749  99.146    nan     nan    nan     nan
	5.265  98.268  6.447  98.706  8.54   99.474    nan     nan    nan     nan
	5.644  98.48   6.992  98.966  9.382  99.74     nan     nan    nan     nan
	6.055  98.682  7.578  99.211 10.299 100.       nan     nan    nan     nan
	6.491  98.884  8.197  99.439    nan     nan    nan     nan    nan     nan
	6.963  99.084  8.852  99.641    nan     nan    nan     nan    nan     nan
	7.457  99.273  9.544  99.836    nan     nan    nan     nan    nan     nan
	7.976  99.452 10.281 100.       nan     nan    nan     nan    nan     nan
	8.516  99.612    nan     nan    nan     nan    nan     nan    nan     nan
	9.081  99.753    nan     nan    nan     nan    nan     nan    nan     nan
	9.663  99.882    nan     nan    nan     nan    nan     nan    nan     nan
	10.283 100       nan     nan    nan     nan    nan     nan    nan     nan
}\tablerange
\begin{figure}
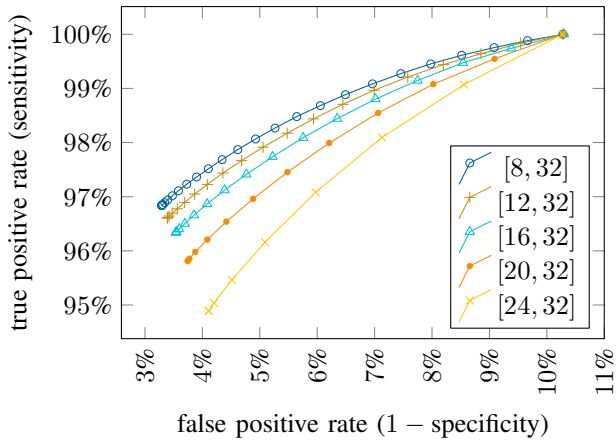

	$$\tikz{\linlin[RoC]{
		\addplot table[x=8fp,y=8sen]{\tablerange};\addlegendentry{$[8,32]$}
		\addplot table[x=12fp,y=12sen]{\tablerange};\addlegendentry{$[12,32]$}
		\addplot table[x=16fp,y=16sen]{\tablerange};\addlegendentry{$[16,32]$}
		\addplot table[x=20fp,y=20sen]{\tablerange};\addlegendentry{$[20,32]$}
		\addplot table[x=24fp,y=24sen]{\tablerange};\addlegendentry{$[24,32]$}
	}}$$
	\global\advance\accumulatenumplots5
	\caption{
		Assume prevalence rate $p=10\%$, uniform Ct values, $15×35$ Kirkman
		triple system, and no delay ($ℓ=0$).  We vary the range of the Ct values
		and plot the ROC curves.  Surprisingly, larger interval (consequently
		larger variance) is easier to decode.
	}\label{fig:ROCrange}
\end{figure}

	\Cref{fig:ROCrange} shows how different ranges of Ct values affect the
	performance.  A larger range leads to a better performance.  We infer that
	this is because more possible Ct values lead to fewer “collisions” and make
	matching easier.  While this could be counterintuitive, it shows that the
	tropical framework is specialized at handling data that span a large range;
	the larger, the better.

\pgfplotstableread{
	8fp    8sen    6fp    6sen    4fp    4sen    2fp    2sen    0fp    0sen
	1.988  97.433  2.326  97.014  2.725  96.596  3.139  96.2    3.549  96.339
	2.023  97.471  2.346  97.039  2.734  96.608  3.142  96.203  3.549  96.339
	2.083  97.535  2.388  97.085  2.763  96.639  3.156  96.221  3.563  96.355
	2.175  97.63   2.458  97.157  2.814  96.692  3.203  96.278  3.611  96.409
	2.303  97.744  2.562  97.259  2.895  96.773  3.3    96.379  3.709  96.508
	2.466  97.884  2.701  97.38   3.057  96.927  3.462  96.549  3.871  96.659
	2.664  98.047  2.873  97.532  3.287  97.134  3.693  96.774  4.104  96.87 
	2.902  98.229  3.18   97.771  3.592  97.385  4.003  97.058  4.409  97.124
	3.18   98.421  3.568  98.049  3.979  97.68   4.389  97.389  4.785  97.42 
	3.559  98.644  4.027  98.357  4.444  98.023  4.854  97.755  5.247  97.758
	3.987  98.87   4.574  98.682  4.983  98.386  5.391  98.139  5.773  98.098
	4.465  99.092  5.077  98.928  5.599  98.75   5.998  98.535  6.368  98.459
	4.989  99.308  5.621  99.155  6.29   99.113  6.676  98.926  7.036  98.814
	5.547  99.513  6.207  99.396  6.884  99.374  7.409  99.277  7.771  99.152
	6.147  99.696  6.823  99.619  7.509  99.613  8.205  99.605  8.554  99.468
	6.792  99.854  7.472  99.821  8.175  99.818  8.857  99.816  9.392  99.741
	7.474 100.     8.16  100.     8.86  100.     9.544 100.    10.306 100.   
}\tablelimit
\begin{figure}
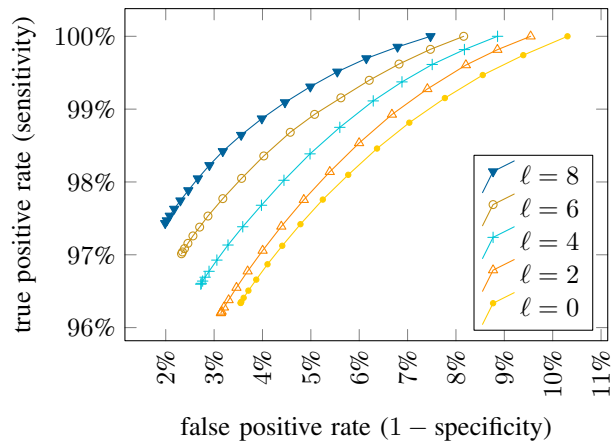

	$$\tikz{\linlin[RoC]{
		\addplot table[x=8fp,y=8sen]{\tablelimit};\addlegendentry{$ℓ=8$}
		\addplot table[x=6fp,y=6sen]{\tablelimit};\addlegendentry{$ℓ=6$}
		\addplot table[x=4fp,y=4sen]{\tablelimit};\addlegendentry{$ℓ=4$}
		\addplot table[x=2fp,y=2sen]{\tablelimit};\addlegendentry{$ℓ=2$}
		\addplot table[x=0fp,y=0sen]{\tablelimit};\addlegendentry{$ℓ=0$}
	}}$$
	\global\advance\accumulatenumplots5
	\caption{
		Assume prevalence rate $p=10\%$, uniform Ct values on the interval
		$[16,32]$, $15×35$ Kirkman triple system, and $ℓ·†Bernoulli†(1/2)$
		delay.  We vary the limit of delay $ℓ$ and plot the ROC curves.
	}\label{fig:ROClimit}
\end{figure}

	\Cref{fig:ROClimit} shows how the delay facilitates matching.  Delaying
	further reduces the collision probability and improves the ROC tradeoff.

\pgfplotstableread{
	Bfp    Bsen     Ufp    Usen     0fp    0sen
	1.99   97.42    2.4    96.751   3.543  96.327
	2.023  97.456   2.415  96.768   3.543  96.327
	2.082  97.519   2.446  96.803   3.556  96.342
	2.173  97.607   2.504  96.864   3.605  96.395
	2.298  97.718   2.601  96.967   3.704  96.496
	2.457  97.854   2.754  97.109   3.861  96.654
	2.659  98.017   2.974  97.3     4.093  96.866
	2.897  98.196   3.272  97.552   4.394  97.122
	3.172  98.386   3.653  97.843   4.772  97.425
	3.555  98.615   4.117  98.178   5.225  97.758
	3.985  98.846   4.646  98.52    5.756  98.116
	4.46   99.068   5.233  98.858   6.356  98.466
	4.977  99.284   5.848  99.173   7.026  98.823
	5.542  99.495   6.479  99.45    7.755  99.162
	6.15   99.682   7.095  99.685   8.527  99.471
	6.793  99.851   7.679  99.863   9.365  99.75 
	7.472 100.      8.193 100.     10.286 100.   
}\tabledistri
\begin{figure}
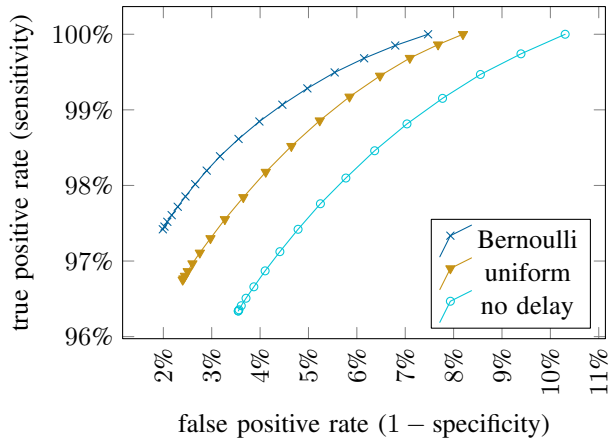

	$$\tikz{\linlin[RoC]{
		\addplot table[x=Bfp,y=Bsen]{\tabledistri};\addlegendentry{Bernoulli}
		\addplot table[x=Ufp,y=Usen]{\tabledistri};\addlegendentry{uniform}
		\addplot table[x=0fp,y=0sen]{\tablelimit};\addlegendentry{no delay}
	}}$$
	\global\advance\accumulatenumplots3
	\caption{
		Assume prevalence rate $p=10\%$, uniform Ct values on the interval
		$[16,32]$, $15×35$ Kirkman triple system, and $ℓ=8$.  We vary the
		distribution of the random delay $δ$ and plot the ROC curves.
	}\label{fig:ROCdistri}
\end{figure}

	\Cref{fig:ROCdistri} shows how the distribution of the random delay is
	correlated to the performance.  Bernoulli (uniform on $\{0,ℓ\}$) is
	apparently better than uniform (uniform on $\{0,\dotsc,ℓ\}$).  We believe
	that this is because the former assumes a greater variance and makes
	collision rarer.

\pgfplotstableread{
	405fp  405sen   105fp  105sen   45fp   45sen    15fp   15sen
	2.588  96.606   2.602  96.61    2.787  96.555   3.534  96.319
	2.588  96.606   2.602  96.61    2.787  96.555   3.534  96.319
	2.601  96.618   2.616  96.622   2.8    96.567   3.547  96.333
	2.652  96.661   2.667  96.666   2.851  96.613   3.595  96.386
	2.759  96.747   2.774  96.75    2.955  96.7     3.695  96.484
	2.935  96.881   2.949  96.885   3.128  96.838   3.854  96.64 
	3.19   97.066   3.204  97.066   3.377  97.026   4.082  96.844
	3.53   97.3     3.544  97.295   3.713  97.261   4.385  97.107
	3.955  97.571   3.973  97.565   4.128  97.537   4.766  97.405
	4.471  97.879   4.49   97.87    4.634  97.839   5.227  97.738
	5.072  98.201   5.091  98.197   5.228  98.177   5.755  98.093
	5.76   98.537   5.776  98.531   5.892  98.515   6.354  98.447
	6.521  98.871   6.533  98.869   6.631  98.857   7.019  98.811
	7.356  99.192   7.368  99.192   7.44   99.181   7.747  99.141
	8.258  99.487   8.269  99.491   8.324  99.487   8.536  99.463
	9.23   99.751   9.237  99.754   9.27   99.746   9.372  99.74 
	10.284 100     10.288 100.     10.296 100.     10.292 100.   
}\tablekirkman
\begin{figure}
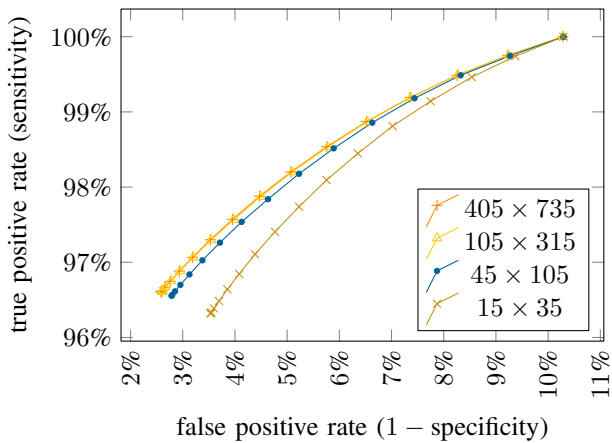

	$$\tikz{\linlin[RoC]{
		\addplot table[x=405fp,y=405sen]{\tablekirkman};\addlegendentry{$405×735$}
		\addplot table[x=105fp,y=105sen]{\tablekirkman};\addlegendentry{$105×315$}
		\addplot table[x=45fp,y=45sen]{\tablekirkman};\addlegendentry{$45×105$}
		\addplot table[x=15fp,y=15sen]{\tablekirkman};\addlegendentry{$15×35$}
	}}$$
	\global\advance\accumulatenumplots4
	\caption{
		Assume prevalence rate $p=10\%$, uniform Ct values on $[16,32]$, and
		no delay ($ℓ=0$).  We consider Kirkman triple systems of different
		size (after truncation so that the code rate $N/T=7/3$ is fixed) and
		plot the ROC curves.
	}\label{fig:ROCkirkman}
\end{figure}
	
	\Cref{fig:ROCkirkman} shows how the size of the matrix can have impact on
	the performance.  At the same code rate $N/T=7/3$, a truncation of a larger
	Kirkman triple system performs better because, presumably, the girth can be
	improved.  That being said, there appeared to be a ceiling on how much the
	girth can help.  All Kirkman systems used here are those returned by the
	\texttt{kirkman\_triple\_system} function in SageMath the mathematical
	software \cite{Sage}.

\pgfplotstableread{
	T0fp    T0sen T1fp   T1sen  T2fp   T2sen  T3fp   T3sen  T4fp   T4sen  T5fp   T5sen  T6fp   T6sen  T7fp   T7sen  T8fp   T8sen  T9fp   T9sen
	1.826  97.984 1.809  97.96  1.823  98.022 1.871  98.007 1.837  97.99  1.857  98.007 1.842  98.054 1.852  98.005 1.866  98.008 1.851  97.983
	1.826  97.984 1.809  97.96  1.823  98.022 1.871  98.007 1.837  97.99  1.857  98.007 1.842  98.054 1.852  98.005 1.866  98.008 1.851  97.983
	1.834  97.99  1.816  97.965 1.833  98.031 1.882  98.011 1.846  97.996 1.869  98.014 1.851  98.056 1.861  98.01  1.875  98.012 1.861  97.989
	1.869  98.004 1.855  97.985 1.867  98.055 1.917  98.033 1.881  98.02  1.904  98.044 1.888  98.067 1.893  98.031 1.909  98.036 1.898  98.008
	1.951  98.04  1.925  98.034 1.949  98.097 1.999  98.082 1.957  98.061 1.978  98.087 1.957  98.122 1.963  98.076 1.979  98.082 1.972  98.058
	2.074  98.116 2.057  98.108 2.08   98.168 2.131  98.146 2.084  98.12  2.114  98.145 2.089  98.185 2.093  98.128 2.112  98.158 2.103  98.128
	2.27   98.209 2.238  98.206 2.269  98.269 2.313  98.236 2.273  98.226 2.295  98.245 2.276  98.272 2.273  98.224 2.3    98.257 2.288  98.22 
	2.535  98.357 2.498  98.352 2.528  98.405 2.562  98.369 2.528  98.347 2.545  98.353 2.523  98.397 2.541  98.351 2.559  98.39  2.548  98.348
	2.857  98.494 2.823  98.521 2.865  98.551 2.894  98.533 2.848  98.509 2.888  98.51  2.86   98.554 2.874  98.507 2.889  98.518 2.878  98.493
	3.284  98.661 3.241  98.704 3.265  98.703 3.305  98.693 3.257  98.67  3.293  98.689 3.282  98.736 3.276  98.683 3.301  98.687 3.286  98.684
	3.761  98.875 3.737  98.906 3.737  98.897 3.785  98.865 3.734  98.85  3.773  98.879 3.767  98.933 3.761  98.873 3.779  98.834 3.776  98.868
	4.31   99.074 4.304  99.097 4.316  99.08  4.343  99.051 4.293  99.052 4.34   99.067 4.318  99.142 4.338  99.072 4.344  99.046 4.357  99.067
	4.964  99.283 4.916  99.3   4.962  99.273 4.98   99.284 4.94   99.235 4.978  99.267 4.98   99.322 4.967  99.263 4.981  99.248 4.992  99.241
	5.674  99.489 5.647  99.493 5.683  99.487 5.709  99.473 5.668  99.425 5.706  99.459 5.687  99.489 5.714  99.483 5.689  99.445 5.712  99.412
	6.481  99.671 6.432  99.702 6.473  99.658 6.525  99.649 6.454  99.607 6.492  99.645 6.482  99.657 6.534  99.66  6.457  99.648 6.513  99.617
	7.353  99.832 7.331  99.867 7.358  99.831 7.356  99.833 7.363  99.791 7.365  99.82  7.369  99.825 7.432  99.845 7.324  99.811 7.388  99.816
	8.331 100.    8.267 100.    8.347 100.    8.297 100.    8.33  100.    8.307 100.    8.315 100.    8.389 100.    8.281 100.    8.339 100.   
}\tabletapestry
\begin{figure}
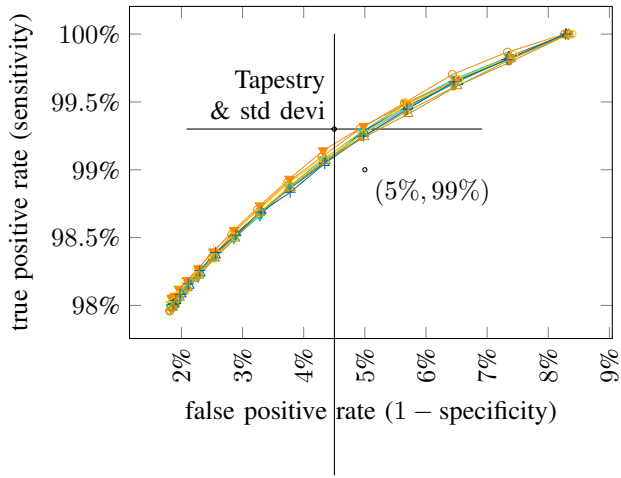

	$$\tikz{\linlin[RoC,every axis plot/.append style={line width=.1}]{
		\draw(4.5,99.3)circle(.8pt)
			node[above left,align=right]{Tapestry \\ \& std devi}
			(4.5+2.41,99.3)--(4.5-2.41,99.3)
			(4.5,100)coordinate(T+y)(4.5,99.3-2.55)coordinate(T-y){}
			(5,99)circle(.8pt)node[below right]{$(5\%,99\%)$};
		\addplot table[x=T0fp,y=T0sen]{\tabletapestry};
		\addplot table[x=T1fp,y=T1sen]{\tabletapestry};
		\addplot table[x=T2fp,y=T2sen]{\tabletapestry};
		\addplot table[x=T3fp,y=T3sen]{\tabletapestry};
		\addplot table[x=T4fp,y=T4sen]{\tabletapestry};
		\addplot table[x=T5fp,y=T5sen]{\tabletapestry};
		\addplot table[x=T6fp,y=T6sen]{\tabletapestry};
		\addplot table[x=T7fp,y=T7sen]{\tabletapestry};
		\addplot table[x=T8fp,y=T8sen]{\tabletapestry};
		\addplot table[x=T9fp,y=T9sen]{\tabletapestry};
		}\draw(T+y)--(T-y);
	}$$
	\global\advance\accumulatenumplots10
	\vskip-12pt
	\caption{
		Assume $D=10$ patients within $N=105$ persons (infection rate
		$9.52\%$), uniform Ct values on $[16,32]$, $45×105$ Kirkman triple
		system (truncation of a $45×330$ Kirkman triple system), and no
		delay ($ℓ=0$).  We plot $10$ ROC curves.  Each curve is $10{,}000$
		encoding--decodings, i.e., $450{,}000$ tubes, $100{,}000$ patients,
		and $1{,}050{,}000$ test takers.  Compare this to Tapestry's data
		point and its standard deviations $(4.50\%±2.41\%,99.30\%±2.55\%)$
		(Table S.XII of the preprint version \cite{GARPAGCGGR20}).
	}\label{fig:ROCtapestry}
\end{figure}

	\Cref{fig:ROCtapestry} shows a comparison of tropical group testing to
	Tapestry \cite{GARPAGCGRGRG21}.  Under the same number of patients ($D=10$),
	the same matrix dimension ($45×105$), and no delay ($ℓ=0$), we can also
	achieve $99\%$ sensitivity and $95\%$ specificity.  What's more, our setting
	is more hazard.  We assume uniform Ct value on $[16,32]$; Tapestry assumes
	uniform viral load on $[1,32768]$.  We assume that $c_t$ is $㏒₂(v_t)$
	rounded up to an integer; Tapestry assumes that $c_t$ is $㏒₂(v_t)$ shifted
	by $0.1Z㏒_2(1.95)≈0.096Z$, where $Z$ is a standard Gaussian.

\pgfplotstableread{
	97fp   97sen    49fp   49sen    17fp   17sen    Pfp    Psen
	0.892  96.809   1.015  96.099   1.862  94.293   2.9    96.004
	0.892  96.809   1.015  96.099   1.862  94.303   2.9    96.004
	0.893  96.809   1.022  96.108   1.922  94.424   2.901  96.005
	0.899  96.815   1.045  96.141   2.037  94.625   2.907  96.013
	0.916  96.832   1.095  96.204   2.203  94.906   2.927  96.037
	0.953  96.869   1.179  96.31    2.416  95.242   2.972  96.091
	1.018  96.934   1.301  96.462   2.671  95.626   3.056  96.191
	1.12   97.037   1.469  96.661   2.963  96.066   3.193  96.347
	1.268  97.179   1.682  96.92    3.297  96.515   3.395  96.565
	1.468  97.378   1.944  97.211   3.659  96.973   3.672  96.852
	1.727  97.629   2.254  97.55    4.047  97.436   4.039  97.206
	2.047  97.936   2.619  97.937   4.463  97.907   4.501  97.631
	2.436  98.294   3.025  98.338   4.9    98.355   5.065  98.096
	2.891  98.69    3.482  98.76    5.362  98.806   5.721  98.588
	3.414  99.12    3.988  99.176   5.834  99.233   6.482  99.093
	4.008  99.557   4.536  99.596   6.327  99.626   7.334  99.561
	4.69  100.      5.135 100.      6.837 100.      8.296 100.   
}\tabelpbest
\begin{figure}
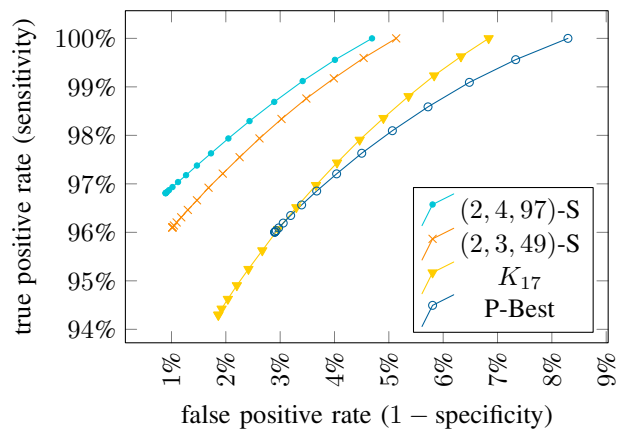

	$$\tikz{\linlin[RoC]{
		\addplot table[x=97fp,y=97sen]{\tabelpbest};\addlegendentry{$(2,4,97)$-S}
		\addplot table[x=49fp,y=49sen]{\tabelpbest};\addlegendentry{$(2,3,49)$-S}
		\addplot table[x=17fp,y=17sen]{\tabelpbest};\addlegendentry{$K_{17}$}
		\addplot table[x=Pfp,y=Psen]{\tabelpbest};\addlegendentry{P-Best}
	}}$$
	\global\advance\accumulatenumplots4
	\caption{
		Assume prevalence rate $p=2\%$, uniform Ct values on $[16,32]$,
		and no delay ($ℓ=0$).  We consider $(2,4,97)$-Steiner system
		(aka $2$-$(97,4,1)$ design), $(2,3,49)$-Steiner system (aka
		$2$-$(49,3,1)$ design), complete graph on $17$ vertices, and P-BEST
		\cite{SLWSSOGSEGGMSFNSPH20}.  They all have code rate $N/T=8$.
		We plot their ROC curves.
	}\label{fig:ROCpbest}
\end{figure}

	\Cref{fig:ROCpbest} shows a comparison of Steiner systems to the block
	design used in P-BEST \cite{SLWSSOGSEGGMSFNSPH20}.  Per the report, P-BEST
	is applicable when the prevalence rate is $<1.3\%$.\footnote{ They claimed
	so because their experiments and simulations show that P-BEST, with high
	probability, can identify $5$ patients among $384$ suspects.  It it unclear
	from their paper what happens when there are more than $5$ patients.}  We
	found that, under tropical group testing, the same matrix performs rather
	good til $p=2\%$.  That being said, there are block designs at the same code
	rate ($N/T=8$) that perform better.  For instance, at $p=2\%$, the complete
	graph on $17$ vertices, which is by definition a $(2,2,17)$-Steiner system,
	has a better sensitivity--specificity tradeoff with fewer vertices and
	simpler pipetting rules.  The $(2,3,49)$-Steiner system (returned by the
	\texttt{steiner\_triple\_system} function in \cite{Sage}) has a better
	tradeoff when $p∈[1\%,2\%]$.  The $(2,4,97)$-Steiner system
	\cite[Theorem~2.2]{RR12} has a better tradeoff across all $p∈[0\%,2\%]$.  

\pgfplotstableread{
	183fp  183sen   15fp   15sen    61fp   61sen
	0.65   91.622   3.018  96.277   1.182  86.561
	0.65   91.622   3.019  96.287   1.182  86.583
	0.656  91.633   3.052  96.356   1.244  86.87 
	0.679  91.679   3.12   96.477   1.365  87.367
	0.727  91.784   3.227  96.664   1.54   88.01 
	0.81   91.995   3.37   96.865   1.766  88.772
	0.931  92.29    3.552  97.125   2.039  89.643
	1.095  92.696   3.775  97.407   2.351  90.582
	1.304  93.19    4.028  97.709   2.709  91.595
	1.564  93.803   4.326  98.038   3.092  92.62 
	1.875  94.507   4.657  98.358   3.508  93.683
	2.234  95.298   5.015  98.689   3.946  94.744
	2.643  96.161   5.405  98.98    4.42   95.845
	3.099  97.062   5.822  99.282   4.921  96.933
	3.611  98.033   6.272  99.567   5.441  98.   
	4.169  99.029   6.744  99.792   5.982  99.028
	4.775 100.      7.236 100.      6.546 100.   
}\tablehyper
\begin{figure}
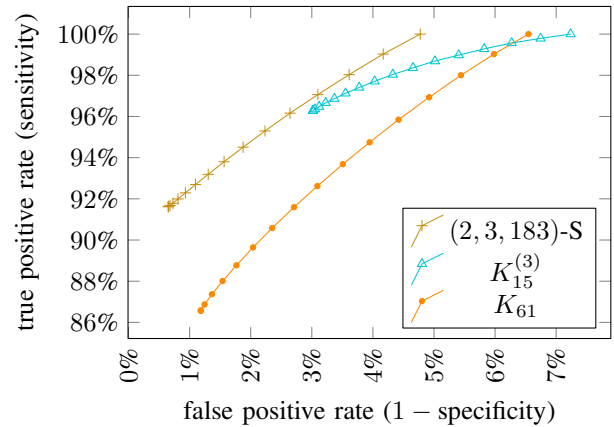

	$$\tikz{\linlin[RoC]{
		\addplot table[x=183fp,y=183sen]{\tablehyper};\addlegendentry{$(2,3,183)$-S}
		\addplot table[x=15fp,y=15sen]{\tablehyper};\addlegendentry{$K^{(3)}_{15}$}
		\addplot table[x=61fp,y=61sen]{\tablehyper};\addlegendentry{$K_{61}$}
	}}$$
	\global\advance\accumulatenumplots3
	\caption{
		Assume prevalence rate $p=0.5\%$, uniform Ct values on $[16,32]$, and no
		delay ($ℓ=0$).  We consider Kirkman triple system on $183$ vertices,
		complete $3$-uniform hypergraph on $15$ vertices, and complete graph on
		$61$ vertices.  The first two have code rate $N/T=30+1/3$; the last one
		has code rate $N/T=30$.  We plot their ROC curves.
	}\label{fig:ROChyper}
\end{figure}

	\Cref{fig:ROChyper} shows a comparison of three designs at $p=0.5\%$.
	Kirkman triple system on $183$ vertices performs better than the other two
	at the cost of reasonable subpacketization---it uses $183$ tests on $5551$
	persons.  Complete $3$-uniform hypergraph on $15$ vertices, using $15$ tests
	on $\binom{15}3=455$ persons, is easier to implement and still achieves
	$98\%$ sensitivity and $95\%$ specificity.  The last one, the complete graph
	on $61$ vertices, uses only two pipets per person and achieves a comparable
	performance.

\hbadness10000
\def\bibsetup{%
	\interlinepenalty=1000\relax
	\widowpenalty=1000\relax
	\clubpenalty=1000\relax
	\frenchspacing
	\biburlsetup
}
\printbibliography

\end{document}